\documentclass[lettersize,journal]{IEEEtran}
\usepackage{amsmath,amsfonts}
\usepackage{algorithmic}
\usepackage{algorithm}
\usepackage{array}
\usepackage[caption=false,font=normalsize,labelfont=sf,textfont=sf]{subfig}
\usepackage{textcomp}
\usepackage{stfloats}
\usepackage{url}
\usepackage{verbatim}
\usepackage{graphicx}
\usepackage{graphics}
\usepackage{cite}
\usepackage{nomencl}
\makenomenclature
\usepackage{enumitem}
\usepackage{comment}
\usepackage{xcolor}
\usepackage{hyperref}
\hypersetup{
colorlinks,
    citecolor=blue,
    linkcolor=red,    
    urlcolor=red,
    }
\usepackage{tabstackengine}
\usepackage{scalerel}
\usepackage{stackengine}
\usepackage{accents}
\usepackage[english]{babel}
\usepackage{amsthm}
\newtheorem{theorem}{Theorem}
\newtheorem{innercustomthm}{Case}
\newenvironment{case}[1]
  {\renewcommand\theinnercustomthm{#1}\innercustomthm}
  {\endinnercustomthm}

\newtheorem{definition}{Definition}
\newtheorem{assumption}{Assumption}

\newtheorem{remark}{Remark}

\newtheorem*{example}{Example}
\newtheorem{lemma}{Lemma}
\newcommand{\ut}[1]{\underaccent{\tilde}{#1}}
\renewcommand{\vec}[1]{\ut{#1}}

\DeclareMathAlphabet{\mymathbb}{U}{BOONDOX-ds}{m}{n}
\begin{document}

\title{
  Adaptive Relaxation based Non-Conservative Chance Constrained Stochastic MPC 
{\footnotesize \textsuperscript{}}
}
\author{Avik~Ghosh, 
        Cristian~Cortes-Aguirre, 
        Yi-An~Chen, 
        Adil~Khurram, 
        and Jan~Kleissl 
        \vspace{-1em}
\thanks{A.~Ghosh \textit{(corresponding author)}, C.~Cortes-Aguirre, Y.~Chen and A.~Khurram are with the Department of Mechanical and Aerospace Engineering, University of California, San Diego, CA, 92093 USA (UCSD) \\
email:\{avghosh,ccortesa,yic002,akhurram\}@ucsd.edu. \\
J.~Kleissl is the Director of the Center for Energy Research in the Department of Mechanical and Aerospace Engineering at UCSD. \\
email: jkleissl@ucsd.edu.}}

\maketitle

\begin{abstract}

Chance constrained stochastic model predictive controllers (CC-SMPC) trade off full constraint satisfaction for economical plant performance under uncertainty. Previous CC-SMPC works are over-conservative in constraint violations leading to worse economic performance. Other past works require a-priori information about the uncertainty set, limiting their application. 
This paper considers a discrete 
LTI system with hard constraints on inputs and chance constraints on states, with unknown uncertainty distribution, statistics, or samples. This work proposes a novel adaptive online update rule to relax the state constraints based on the time-average of past constraint violations, 
to achieve reduced conservativeness in closed-loop. Under an ideal control policy assumption,
it is proven that the time-average of constraint violations asymptotically converges to the maximum allowed violation probability. 
The 
method is applied for optimal battery energy storage system (BESS) dispatch in a grid connected microgrid with PV generation and load demand, with chance constraints on BESS state-of-charge (SOC). \textcolor{black}{Realistic simulations show the superior electricity cost saving potential of the proposed method as compared to the traditional economic MPC without chance constraints, 
and a state-of-the-art approach 
with chance constraints. We satisfy the chance constraints non-conservatively in closed-loop, effectively trading off increased cost savings with minimal adverse effects on BESS lifetime.}
\end{abstract}


\nomenclature{VRES}{Variable renewable energy sources}
\nomenclature{PV}{Photovoltaic}
\nomenclature{BESS}{Battery energy storage system}
\nomenclature{MPC}{Model predictive control}
\nomenclature{SMPC}{Stochastic MPC}
\nomenclature{MG}{Microgrid}
\nomenclature{JCC}{Joint chance constraints}
\nomenclature{NCDC}{Non-coincident demand charge}
\nomenclature{OPDC}{On-peak demand charge}
\nomenclature{\textcolor{black}{NCDP}}{\textcolor{black}{Non-coincident demand peak}}
\nomenclature{\textcolor{black}{OPDP}}{\textcolor{black}{On-peak demand peak}}
\nomenclature{SOC}{State-of-charge}
\nomenclature{LTI}{Linear time invariant}
\nomenclature{$x$}{State}
\nomenclature{$u$}{Control input}
\nomenclature{$w$}{Uncertainty}
\nomenclature{$\hat w$}{Width of the critical region}
\nomenclature{$s$}{Control input constraint vector}
\nomenclature{$m$}{Control input dimension}
\nomenclature{$g$}{State chance constraint vector}
\nomenclature{\textcolor{black}{$c$}}{\textcolor{black}{Control input coupling vector}}
\nomenclature{$A$}{System state transition matrix}
\nomenclature{$B$}{System control input matrix}
\nomenclature{$E$}{System state uncertainty matrix}
\nomenclature{$F$}{Control input uncertainty matrix}
\nomenclature{$\mathbb{P}$}{Probability measure}
\nomenclature{$\mathcal{F}_t$}{Filtration}
\nomenclature{$S$}{Control input constraint matrix}
\nomenclature{$M$}{Control input coupling matrix}
\nomenclature{$G$}{State chance constraint matrix}
\nomenclature{$\alpha$}{Maximum probability of violation of state constraints}
\nomenclature{$N$}{MPC prediction horizon length}
\nomenclature{$\tilde h$}{Adaptive state constraint tightening parameter}
\nomenclature{$h$}{Adaptive state constraint relaxing parameter}
\nomenclature{$k, t$}{Time index}
\nomenclature{$V$}{State constraint violation tracker}
\nomenclature{$Y$}{Time-average of state constraint violations}
\nomenclature{$Z$}{Absolute difference between $\alpha$ and $Y$}
\nomenclature{$\gamma$}{Constant of proportionality in online $h$ update rule}
\nomenclature{$\kappa$, $\kappa'$}{Critical region}
\nomenclature{$n$}{State dimension}
\nomenclature{$p$}{Uncertainty dimension / Probability of violation of state constraints}
\nomenclature{$q$}{Control input constraint vector dimension}
\nomenclature{$d$}{Control input coupling vector dimension}
\nomenclature{$r$}{State chance constraint vector dimension}
\printnomenclature
    
\section{Introduction}\label{intro}
\subsection{Motivation}\label{motivation}

Currently, there is great emphasis on integrating variable renewable energy sources (VRES), such as wind and PV generators into the electricity grid, with the goal of de-carbonizing power production. There is, however, an intermittent nature to VRES, which 
can potentially lead to power imbalance in the electric grid, thereby risking grid stability~\cite{Kou2015}. Battery energy storage systems (BESS) can minimize power fluctuations caused by the integration of VRES into the grid~\cite{Kou2015}, and can additionally be used for energy arbitrage, peak load shaving, valley-filling, and ancillary services. 
However, to maximize benefits from installing BESS, optimal BESS scheduling strategies need to be devised to maximize electricity bill savings, while providing services to the grid.

It is possible to optimally dispatch BESS, utilizing Model predictive control (MPC) based scheduling algorithms that include grid constraints. 
However, uncertainty in forecasts can significantly reduce performance and should be taken into account when formulating MPCs.
Classical open-loop min-max formulation based robust MPC can be used to factor in uncertainty but it leads to over-conservative solutions which may not be economical from an operational perspective~\cite{Stanford}. Other variations of robust MPC such as closed-loop min-max formulation (commonly known as ``Feedback MPC") suffer from prohibitive complexity~\cite{Stanford}. Tube-based MPC requires specification of a bounded uncertainty set a-priori~\cite{IEEE_1} (a problem in common with robust MPC), which may be difficult to specify \textcolor{black}{non-conservatively} for a complex practical system such as a VRES integrated MG, which involves a variety of forecasts. \textcolor{black}{An over-conservative uncertainty set negatively impacts economical system performance.} 

Stochastic MPC (SMPC) methods based on chance constraints strike a trade off between economic operation and full constraint satisfaction~\cite{IEEE_1}. Chance constraints allow for the MPC to operate in a more economical way by respecting a maximum probability of constraint violation. The superior economic performance, lower complexity, and weaker assumption requirements of chance constrained SMPC is desirable for BESS operation in VRES intensive microgrids (MG), especially under uncertainty in VRES and load forecasts, and is thus the primary focus of this work.

 \vspace{-0.7em}
\subsection{Literature Review}\label{lit_review} 

Chance constrained SMPC algorithms have found applications in problems involving building climate control~\cite{IEEE_1,munoz2018stochastic,long2020iterative,korda2014stochastic}, optimal power flow~\cite{IEEE_2}, and optimal microgrid (MG) dispatch~\cite{guo2017model,Liu2017,Yuan2021,Garifi2018,Gulin2015,ding2022distributionally,aghdam2020stochastic,wang2023optimal, ciftci2019data, Guo2018, ghosh2023adaptive}. Chance-constrained SMPC problems are solved by converting them into an approximate deterministic form. If the uncertainties are Gaussian, or follow other known distributions~\cite{special_distribution_1,special_distribution_2}, standard procedures exist to convert the stochastic problem into a deterministic one. \textcolor{black}{However, in practical scenarios such as VRES and load forecasts, the uncertainty distributions, and additionally, uncertainty statistics (moments like mean, variance, skewness, etc.) may be unknown and can vary with time (i.e., seasonally/yearly).} Other methods of reformulating the SMPC problem into a deterministic one, such as using Chebyshev inequalities~\cite{chebyshev} require a-priori knowledge about the uncertainty statistics (mean and covariance), while using Chernoff bounds suffer from high conservatism~\cite{IEEE_1}. Sampling based approaches~\cite{scenario_1,scenario_2} suffer from high computational demand and may require a prohibitive number of samples~\cite{IEEE_1}. 

Chance constraints are generally enforced pointwise-in-time within an MPC prediction horizon, without including past behavior of the system, which can also lead to over-conservativeness (i.e., less than desired constraint violations) in closed-loop \cite{korda2014stochastic,munoz2018stochastic}. However, reducing over-conservativeness is of paramount importance for MG operators to reduce electricity costs. \textcolor{black}{Thus, in this work we re-interpret the pointwise-in-time chance constraints as time-average of violations (or time-average of some loss function of violations) in closed-loop similar to~\cite{IEEE_1,munoz2018stochastic,time_avg,korda2014stochastic,capone2024online}, and first focus our literature review specifically on such online chance constrained SMPC methods with theoretical advancements. The re-interpretation keeps the spirit of occasional constraint violations of the original chance constraint\cite{munoz2018stochastic} while keeping memory of past behavior of the system to aid in reducing over-conservativeness in closed-loop. }


\textcolor{black}{The work in~\cite{korda2014stochastic} 
adaptively relaxed and tightened the MPC state constraints online aided by the amount of violations quantified by a loss function empirically weighted averaged over time. }
The authors defined a family of stochastic robust control invariant (SRCI) sets for implementing their control online and proved that the empirical weighted average loss is bounded either in expected value or robustly with probability 1, and derived bounds on the convergence time. However, drawbacks of the work include high computational cost to parameterize the SRCI sets, and a-priori knowledge about the distribution/statistics of the uncertainty. 

\textcolor{black}{The authors in~\cite{time_avg} 
adaptively relaxed the MPC state constraints online 
based on the time-average of (i) the number of constraint violations in the first method, and (ii) a loss metric based on a convex loss function of constraint violations in the other method.} Instead of providing asymptotic bounds on constraint violations as done in previous works like~\cite{IEEE_1,munoz2018stochastic}, the authors provided stronger robust bounds in closed-loop over finite time periods.
A practical limitation of~\cite{time_avg} for our application (economic MG dispatch) is the assumption of an objective function that is composed of stage costs with quadratic penalties on predicted control inputs and state deviations from an a-priori defined robust positively invariant target set. Construction of a non-conservative robust positively invariant target state set is difficult. Additionally, electricity costs due to grid imports in a MG with BESS cannot be expressed by stage costs with quadratic penalties on control inputs and state deviations, because economic stage costs are not necessarily positive definite with respect to the target set of states and/or control inputs~\cite{rawlings2012fundamentals}.\footnote{Economic MG dispatch involves usage of the electricity cost function directly as the objective function of the MPC controller. Electricity costs incurred by the MG to the utility involve time-of-use energy charges and demand charges. Energy charges (\$/kWh) are incurred based on the volumetric import of electricity from the grid, while demand charges (\$/kW) are incurred based on the maximum load import from the grid over the month. \textcolor{black}{Two distinct time-of-use demand charges are used: one based on maximum grid imports for the whole month called non-coincident demand peak (NCDP), added on top of maximum grid import between 16:00-21:00 h of all days of the month, called on-peak demand peak (OPDP). The demand charges associated with NCDP and OPDP are called non-coincident demand charges (NCDC), and on-peak demand charges (OPDC) respectively.} For commercial and industrial customers, demand charge costs are typically 30-70\% of the monthly electricity costs~\cite{Demand_charges}.} \textcolor{black}{Moreover, the uncertainty is incorporated by bounding the predicted states in state tubes which require a-priori specification of the uncertainty set, and the computation time is similar to a corresponding robust MPC problem which is more expensive than the nominal MPC.}


\textcolor{black}{The authors in~\cite{IEEE_1,munoz2018stochastic,capone2024online} 
  adaptively tightened the state constraints online during MPC computations
  : (i) using an update rule based on the time-average of past state constraint violations in~\cite{IEEE_1,munoz2018stochastic}, 
  and (ii) by iteratively employing a data-driven Gaussian process binary regression based approach depending on the observed state constraint violations in the training data in~\cite{capone2024online}. 
  In \cite{munoz2018stochastic}, using stochastic approximation, the authors also proved the convergence of the time-average of constraint violations in probability to the allowable `least conservative' level. A few limitations of \cite{munoz2018stochastic} are the a-priori assumption of the uncertainty distribution, and the strong assumption of terminal stability region of the state in closed-loop, which is unrealistic for economic MG dispatch using BESS. While~\cite{capone2024online} relaxes the requirement of a-priori knowledge of the uncertainty distribution, iteratively learning the final optimal tightening parameter is data-intensive. The entire solution framework exhibits significantly more computation cost for satisfying the chance constraints in the long run 
  as compared to the nominal MPC. 
  For our application,~\cite{capone2024online} may be unable to perform economically or may cause significant violation of the chance constraints if the underlying uncertainty distribution changes with time (as the final optimal tightening parameter is learned from past data), which is undesirable for economic MG operation.} 

\textcolor{black}{Some of the limitations of \cite{korda2014stochastic,time_avg, munoz2018stochastic,capone2024online} are avoided by~\cite{IEEE_1} which can be applied to systems with unknown uncertainty distribution/statistics. Also,~\cite{IEEE_1} does not need the assumptions of terminal stability of the state and the type of the objective function (except convexity assumptions), and the implementation of the SMPC is computationally inexpensive and has similar computation cost as that of a nominal MPC. The above mentioned properties make the work in~\cite{IEEE_1} ideal for economic MG dispatch, and forms the primary reference on which we develop our work.} 

In \cite{IEEE_1}, the authors developed an adaptive state constraint tightening rule that allows the time-average of violations of the system to converge to the maximum allowable violation probability 
in closed-loop under an ideal control policy assumption (which may be unmet for practical implementation). 
However, as clarified by the authors, the convergence was argued intuitively without rigorous proof. 
\textcolor{black}{Also, despite its advantages,~\cite{IEEE_1} is still not robust to significant violation of chance constraints under time-varying uncertainty distribution.
}

While significant theoretical advancements have been made to develop SMPC with varying degrees of simplifying assumptions, computational cost, and ability to avoid over-conservativeness in closed-loop chance constraint satisfaction, there exists a gap in the literature of exploiting these non-conservative methods for economic MG dispatch under uncertainty.  
\textcolor{black}{Some recent applied works involving economic MG dispatch in VRES intensive grids with chance constraints are explored in~\cite{guo2017model,Liu2017,Yuan2021,Garifi2018,Gulin2015,ding2022distributionally,aghdam2020stochastic,wang2023optimal, ciftci2019data}.}

None of the works~\cite{guo2017model,Liu2017,Yuan2021,Garifi2018,Gulin2015,ding2022distributionally,aghdam2020stochastic,wang2023optimal, ciftci2019data} considered demand charges in their electricity cost. The work in~\cite{guo2017model} required extra steps to generate scenarios of uncertainty at every time-step for the MPC prediction horizon and considered the scenarios to be Gaussian. The works~\cite{Liu2017,Garifi2018,Gulin2015,aghdam2020stochastic,wang2023optimal} considered the uncertainties to be Gaussian or other commonly known distributions, which significantly limits the practical application of these studies. \textcolor{black}{
The authors in~\cite{Yuan2021,ding2022distributionally} employed ambiguity sets to model the uncertainties from historical data, and employed distributionally robust chance constrained optimization (DRCC). The authors in~\cite{ciftci2019data} used adaptive kernel density estimation to estimate the nonparametric uncertainty distribution for VRES from historical data, and adjusted the confidence levels according to estimated uncertainties to ensure constraint satisfaction within predefined confidence levels.} However, the approaches in ~\cite{Yuan2021,ciftci2019data,ding2022distributionally} are data-intensive and the performance is substantially influenced by the quality and volume of historical data available.

%
To tackle the aforementioned problems,~\cite{Guo2018} and the authors of this paper in a previous work~\cite{ghosh2023adaptive} presented an online adaptive SMPC model inspired by~\cite{IEEE_1}. \textcolor{black}{The authors in~\cite{Guo2018,ghosh2023adaptive} minimized the VRES integrated MG operating cost, and satisfied chance constraints on states (BESS SOC) 
in closed-loop without making any assumption about the probability distribution or statistics of the uncertainty.} To further reduce over-conservativeness, 
the works~\cite{ghosh2023adaptive,Guo2018} employed online adaptive constraint relaxation in the nominal MPC. 
However, the online adaptive constraint relaxation rules used in~\cite{Guo2018} and~\cite{ghosh2023adaptive} are based on intuition having no convergence guarantees or theoretical analyses, and both the works are application specific. From an economic MG dispatch perspective, only~\cite{ghosh2023adaptive} included demand charges while~\cite{Guo2018} did not. The authors in~\cite{Guo2018} used sample historical data of uncertainties for initial constraint relaxation which~\cite{ghosh2023adaptive} avoided by practical engineering approximation. \textcolor{black}{Additionally,~\cite{Guo2018,ghosh2023adaptive}
considered aggregate constraint violations, without any preference for the time when violations should occur. For economic MG dispatch with demand charges, it becomes critical, if necessary, to preferentially be able to violate BESS state constraints during a predefined on-peak period from 16:00 to 21:00~h, to reduce grid import power peaks, as OPDC 
are charged on top of NCDC
.}

\textcolor{black}{The present work is an extension of the previous work~\cite{ghosh2023adaptive} by the authors, presented in a generic discrete linear time invariant (LTI) setting with additional convergence properties and proofs, related theoretical analyses, and additional case studies.
%
The proposed online adaptive SMPC (OA-SMPC) minimizes a generic convex cost function over a finite receding horizon, subject to hard input constraints and 
chance constraints on states. 
After presenting the theoretical results of the OA-SMPC, a case study is presented for a grid-connected MG operation with PV, load, and BESS using realistic data for a full year of operation in an economic MPC (EMPC) framework. The performance of the OA-SMPC is compared with a traditional EMPC without chance constraints, and a state-of-the-art approach from the literature with chance constraints having similar computational cost~\cite{IEEE_1}. The OA-SMPC outperforms both the methods with respect to the cost saving potential and non-conservative satisfaction of chance constraints.}

\vspace{-0.5em}
\subsection{Contributions}\label{contributions}
The contributions of the present work are as follows:
\begin{enumerate}
    \item 
    \textcolor{black}{To the best of the author's knowledge, 
    the present work's 
    adaptive state constraint relaxation framework is limited in the literature as compared to the more common adaptive state constraint tightening.} Under the novel adaptive relaxation rule of the present work, 
    it is proven that the time-average of the constraint violations asymptotically converges to the maximum allowable violation probability under an ideal control policy assumption similar to~\cite{IEEE_1}. However, a rigorous proof is provided here which was not provided in~\cite{IEEE_1}. \textcolor{black}{Also, for practical implementation (i.e., without the simplifying ideal control policy assumption), while the proposed method cannot guarantee the aforementioned convergence, it still encourages it.}
    
    
    \item \textcolor{black}{The present work also 
    proves that the 
    the time-average of the constraint violations exhibits martingale-like behavior asymptotically
    for practical implementation
    .}

    \item The present work does not require either any a-priori assumption about the probability distribution of the uncertainty set or its statistics, or sample uncertainties from historical data. \textcolor{black}{The present work is also robust to significant violation of chance constraints under time-varying uncertainty distribution for practical implementation, provided an additional post-processing step is incorporated.}
    
    \item \textcolor{black}{The present work incorporates operational adjustments in the online adaptive relaxation rule to account for temporal preference in state constraint violation, which is critical for economic MG dispatch.} Additionally, unlike the methods in~\cite{Guo2018,IEEE_1}, the present method prevents excessive overcharging/overdischarging of the BESS to correct for large forecast uncertainties in real-time by the post-processing step, which otherwise might harm the BESS and leave the MG vulnerable for future demand peaks.
    
    \item The majority of the earlier works for chance constrained SMPC based economic MG dispatch 
    presented results over a short time (24~h), or one or two months. 
    However, for MG operators, it is important to have at-least year-long studies to determine how the algorithm performs under realistic seasonal variations in loads, VRES generation, and forecasts, a gap which the present work fills. 
\end{enumerate}
The rest of the paper is organized as follows.
Section~\ref{prelim} presents the notations and standard definitions. Section~\ref{problem} introduces the original SMPC problem formulation with chance constraints, approximates the original formulation to frame the OA-SMPC formulation, along with presenting the convergence proofs and related theoretical analysis. Section~\ref{case_study} presents the case study for a realistic grid connected MG with PV, load, and BESS, with results and discussions in Section~\ref{results}. Section~\ref{conclusions} concludes the paper summarizing the takeaways of the study.


\vspace{-0.5em}
\section{Mathematical preliminaries}\label{prelim}
\subsection{Notations}\label{notations}

The set of $n$-tuple of real numbers is denoted by $\mathbb{R}^{n}$. Positive and negative real number sets are denoted by $\mathbb{R}_{0+}$ and $\mathbb{R}_{0-}$, respectively. The set of natural numbers including 0 is denoted by $\mathbb{N}$. 
A set of consecutive natural numbers $\{i,i\!+\!1,\ldots,j\}$ is denoted by $\mathbb{N}_{i}^{j}$. The $n$-tuple of ones is denoted by $\mymathbb{1}^n$. States, control inputs and uncertainties are denoted by $x \in \mathbb{X} \subseteq {R}^n$, $u \in \mathbb{U} \subset \mathbb{R}^m$, and $w \in \mathbb{W} \subset \mathbb{R}^p$ respectively. The prediction horizon of the MPC is denoted by $N \in \mathbb{N}$. Actual states at time $t \in \mathbb{T} \subseteq \mathbb{N}$ are denoted by $x(t)$, while predicted states, obtained at $t$ by MPC computation $k \in \mathbb{N}_1^N$ time-steps in the future are denoted by $x(t\!+\!k|t)$. Similarly, predicted control inputs over the MPC prediction horizon are denoted by $u(t\!+\!k|t)$, with $k \in \mathbb{N}_0^{N-1}$. An ordered collection of vectors (such as states) over the MPC prediction horizon obtained at time $t$ is denoted by bold letters, $\textbf{\text x}(t+1) := \Bigl(x(t\!+\!1|t),x(t\!+\!2|t),\ldots,x(t\!+\!N|t)\Bigr)$. For matrices $A$ and $B$ of equal dimensions, the operators $\{<,\leq,=,>,\geq\}$ hold component wise. The right inverse of a matrix $A \in \mathbb{R}^{m \times n}$ with rank $m<n$ is denoted by $A^{\dagger}$.
The $i^{\rm th}$ row, and the element from the $i^{\rm th}$ row and $j^{\rm th}$ column of a matrix $A$ is denoted by $A_i$ and $A_{ij}$ respectively, while the $i^{\rm th}$ element of a vector $x$ is denoted by $x_i$, unless mentioned otherwise. \textcolor{black}{A vector of the first $a \in \mathbb{N}$ elements of a vector $x$ is denoted by $x_{1:a}$.} The expected value of a random variable $Z$ is denoted by $\mathbb{E}[Z]$. $|x|$ denotes the 1-norm of a vector $x$. \textcolor{black}{The `logical not', and `logical and' operators are denoted by $\neg$ and $\land$ respectively.} 

\subsection{ Standard Definitions}\label{definitions}

\begin{definition}[\textbf{Filtered probability space~\cite{williams1991probability}}]\label{prob_space}
A filtered probability space is defined by $(\Omega, \mathcal{F},\{\mathcal{F}_t\}, \mathbb{P})$. $(\Omega, \mathcal{F}, \mathbb{P})$ is a probability triple with sample space $\Omega$, $\sigma$-algebra (event space) $\mathcal{F}$, and probability measure $\mathbb{P}$ on 
$(\Omega, \mathcal{F})$. $\mathcal{F}_t$ is a filtration, which is an increasing family of sub $\sigma$-algebras of $\mathcal{F}$ such that $\mathcal{F}_s \subseteq \mathcal{F}_t \subseteq \mathcal{F}$, $\forall t\geq s$, where $t,s \in \mathbb{T}$.  
\end{definition}

\begin{definition}[\textbf{Almost surely~\cite{williams1991probability}}]\label{as}
An event $E\in\mathcal{F}$ happens almost surely if $\mathbb{P}(E)=1$. It is denoted by a.s.
\end{definition}

\begin{definition}[\textbf{Adapted stochastic process~\cite{williams1991probability}}]\label{adapted}
A stochastic process $Z :=(Z(t) : t> 0)$, is called adapted to the filtration $\{\mathcal{F}_{t}\}$ if $Z(t)$ is $\mathcal{F}_t$ measurable $\forall t$.
\end{definition}

\begin{definition}[\textbf{Supermartingale~\cite{williams1991probability}}]\label{stochastic}
A stochastic process $Z$ is called a discrete-time supermartingale relative to $(\{\mathcal{F}_t\}, \mathbb{P})$ if it satisfies the following:
\begin{enumerate}[label=(\alph*)]
    \item $Z$ is an adapted process,
    \item $\mathbb{E}[|Z(t)|] < \infty$, $\forall t$,
    \item $\mathbb{E}[Z(t+1)|\mathcal{F}_t] \leq Z(t)$, a.s. $\forall t$.
    \end{enumerate}
    \textcolor{black}{A discrete-time martingale $Z$ relative to $(\{\mathcal{F}_t\}, \mathbb{P})$ is defined similarly, with (c) replaced by $\mathbb{E}[Z(t+1)|\mathcal{F}_t] = Z(t)$, a.s. $\forall t$.}
    
\end{definition}

\begin{definition}[\textbf{Monotone convergence theorem for decreasing sequence\cite{bartle1964elements}}]\label{monotone}
Let $X=(x_n: n\in \mathbb{N})$ be a sequence of real numbers which is monotonically decreasing in the sense that $x_{n+1}\leq x_n, \forall n$, then the sequence converges if and only if it is bounded, and in which case $\lim_{n\to\infty}  x_n=\inf \{x_n\}$. 
\end{definition}

\section{Problem formulation}\label{problem}
\subsection{System Description}\label{original_system}

The dynamics of the discrete LTI system are governed by
\begin{equation} \label{dynamics_real}
\begin{aligned}
&x(t+1) = Ax(t)+Bu(t)+Ew(t), \qquad \forall t,
\end{aligned}
\end{equation}
where $A \in \mathbb{R}^{n \times n}$, $B \in \mathbb{R}^{n \times m}$, and $E \in \mathbb{R}^{n \times p}$.

\begin{assumption}[\textbf{System}]\label{assumption_system}
(a) 
At each time $t$, a measurement of the state is available. (b) The set of admissible control inputs $\mathbb{U}$ and states $\mathbb{X}$ are polytopes containing the origin.   
\end{assumption}

\begin{assumption}[\textbf{Uncertainties}]\label{assumption_uncertainty}
The set of uncertainties $\mathbb{W}$ is bounded and contains the origin.   
\end{assumption}


\noindent In this setup, $\mathcal{F} = \sigma (\{w:w(t) \in \mathbb{W}\}:t \in \mathbb{T})$, and $\mathcal{F}_t = \sigma (\{w(s):w(s) \in \mathbb{W}\}:s < t)$. The system is subject to hard control input constraints and chance constraints on states. The control input constraints are formulated as,
\begin{equation} \label{input hard_real}
\begin{aligned}
&Su(t) \leq s,  \qquad \forall t,
\end{aligned}
\end{equation}
where $S \in \mathbb{R}^{q \times m}$, $s \in \mathbb{R}^{q}$. The time-varying equality constraints coupling the control inputs are formulated as,
\begin{equation} \label{equal_real}
\begin{aligned}
&Mu(t) = \textcolor{black}{c(t)}+Fw(t),  \qquad \forall t,
\end{aligned}
\end{equation}
 where $M \in \mathbb{R}^{d \times m}$, $\textcolor{black}{c} \in \mathbb{R}^{d}$, and $F \in \mathbb{R}^{d \times p}$. In previous works like \cite{IEEE_1}, \eqref{equal_real} is not considered but for applications such as economic MG dispatch, \eqref{equal_real} is important for incorporating physical constraints such as power balance of the MG with the main grid (discussed in detail in Section \ref{post-process}). However, if \eqref{equal_real} is considered in the problem formulation, the $Ew(t)$ term in the RHS of~\eqref{dynamics_real} is dropped as the $Fw(t)$ term in the RHS of \eqref{equal_real} accommodates the uncertainty.\footnote{The $E$ matrix is still required for assigning a unique control input after accommodating the uncertainty in closed-loop for multi-input systems. See details in Section \ref{post-process}.} Additionally, note that constraint \eqref{equal_real} is application specific and is independent of the method and theoretical results presented in this paper. The chance constraints on the states are formulated as,
\begin{equation} \label{chance_real}
\begin{aligned}
&\mathbb{P}[Gx(t) \leq g]\geq \bar{1}-\bar{\alpha},  \qquad \forall t,
\end{aligned}
\end{equation}
 where $G \in \mathbb{R}^{r \times n}$, $g \in \mathbb{R}^{r}$ and $\bar{1} = \mymathbb{1}^{r}$ for individual chance constraints. \textcolor{black}{ $\bar{\alpha}=[\alpha_1, \dots, \alpha_r]^\top$} is the vector of the pointwise-in-time maximum probability of constraint violation, where $\alpha_i \in (0,0.5) \; \forall i \in  \mathbb{N}_{1}^{r}$.
 In the individual chance constraint form,~\eqref{chance_real} can be expressed as, $\mathbb{P}[G_ix \leq g_i]\geq 1-\alpha_i,  \quad \forall i \in  \mathbb{N}_{1}^{r}$. 
 In the joint chance constraint (JCC) form, a single violation probability denoted by ${\alpha}\in (0,0.5)$ can be defined for simultaneous satisfaction of all state constraints as, 
\begin{align*}\textcolor{black}{
    \mathbb{P}[
G_1x \leq g_1\land G_2x \leq g_2\land \dots\land G_rx \leq g_r
]\geq 1-{\alpha}.}
\end{align*}
\noindent 
Note that in this work, we re-interpret the maximum probability of violation of state constraints pointwise-in-time given by the chance constraints~\eqref{chance_real} as the maximum time-average of state constraint violations in closed-loop similar to~\cite{IEEE_1,munoz2018stochastic,time_avg,korda2014stochastic,capone2024online}. $Gx(t) \leq g$ is referred to as the original constraint with respect to which violations are measured.
 
\vspace{-0.5em}
\subsection{Online Adaptive SMPC (OA-SMPC)} \label{oasmpc}

Over the MPC prediction horizon $N$, computed from time $t$, we define the ordered collection of states, control inputs, uncertainties and coupling vectors as,
\begin{equation}
\begin{aligned}
\textbf{\text x}(t+1) &:= \Bigl(x(t\!+\!1|t),x(t\!+\!2|t),\ldots,x(t\!+\!N|t)\Bigr)^\top \in \mathbb{R}^{Nn}, \nonumber \\
\textbf{\text u}(t) &:= \Bigl(u(t|t),u(t\!+\!1|t),\ldots,u(t\!+\!N\!-\!1|t)\Bigr)^\top \in \mathbb{R}^{Nm}, \nonumber \\
\textbf{\text w}(t) &:= \Bigl(w(t|t),w(t\!+\!1|t),\ldots,w(t\!+\!N\!-\!1|t)\Bigr)^\top \in \mathbb{R}^{Np}, \nonumber \\
\textcolor{black}{\textbf{\text c}(t)} &:= \textcolor{black}{\Bigl(c(t|t),c(t\!+\!1|t),\ldots,c(t\!+\!N\!-\!1|t)\Bigr)}^\top \in \mathbb{R}^{Nd}. \nonumber 
\end{aligned}
\end{equation}

\noindent The system dynamics can be written in expanded form as
\begin{equation} \label{dynamics_real_2}
\begin{aligned}
x(t\!+\!k|t) = &A^kx(t|t)+\sum_{i=0}^{k-1}A^{k-1-i}Bu(t\!+\!i|t)+ \\
&\sum_{i=0}^{k-1}A^{k-1-i}Ew(t\!+\!i|t), \qquad \forall k \in \mathbb{N}_{1}^{N},\forall t,
\end{aligned}
\end{equation}
where $x(t|t)=x(t)$. Writing~\eqref{dynamics_real_2} in compact form yields,
\begin{equation} \label{dynamics_compact}
\begin{aligned}
\textbf{\text x}(t+1) = &\textbf{\text A}x(t)+\textbf{\text B}\textbf{\text u}(t) + \textbf{\text E}\textbf{\text w}(t), \qquad \forall t,
\end{aligned}
\end{equation}
where $\textbf{\text A} \in \mathbb{R}^{Nn \times n}$, $\textbf{\text B} \in \mathbb{R}^{Nn \times Nm}$ and $\textbf{\text E} \in \mathbb{R}^{Nn \times Np}$.


\noindent The hard control input constraints over the MPC prediction horizon are formulated as,
\begin{equation} \label{input hard_mpc}
\begin{aligned}
Su(t\!+\!k|t) \leq s,  \qquad \forall k \in \mathbb{N}_{0}^{N-1},\forall t.
\end{aligned}
\end{equation}

\noindent The equality constraints coupling the control inputs over the MPC prediction horizon are formulated as,
\begin{equation} \label{equality_mpc}
\begin{aligned}
Mu(t\!+\!k|t) =\textcolor{black}{c}(t\!+\!k|t)+Fw(t\!+\!k|t),  \qquad \forall k \in \mathbb{N}_{0}^{N-1},\forall t.
\end{aligned}
\end{equation}

\noindent Generally, 
the chance constraints in~\eqref{chance_real} are interpreted for the MPC prediction horizon pointwise-in-time by~\eqref{chance_mpc_previous_1}, which is over-conservative 
in closed-loop (see Remark~\ref{Approx_Chance_Real}). The corresponding relaxed deterministic reformulation of~\eqref{chance_mpc_previous_1} as implemented in the MPC prediction horizon by some previous works~\cite{IEEE_1,munoz2018stochastic} is given by~\eqref{chance_mpc_previous_2}. The aim of the deterministic reformulation is to tighten the state constraints under nominal MPC computations (resulting from ignoring uncertainties, i.e., $\textbf{\text w}(t) = \textbf{0}$) by an adaptive tightening parameter $\tilde h \in \mathbb{R}^{r}$ given by,
\begin{subequations} \label{chance_mpc}
\begin{equation} \label{chance_mpc_previous_1}
\begin{aligned}
&\mathbb{P}[Gx(t\!+\!k|t) \leq g]\geq \bar{1}-\bar{\alpha},  \qquad \forall k\in \mathbb{N}_{1}^{N},\forall t,
\end{aligned}
\end{equation}
\begin{equation} \label{chance_mpc_previous_2}
\begin{aligned}
&Gx(t\!+\!k|t) \leq g-\tilde h(t\!+\!k|t),  \qquad \forall k\in \mathbb{N}_{1}^{N},\forall t, \textbf{\text w}(t) = \textbf{0},
\end{aligned}
\end{equation}
\end{subequations}
where $\tilde h_i >0$, $\forall i \in  \mathbb{N}_{1}^{r}$ and is updated based on the time-average of past state constraint violations in closed-loop. Note that violations of state constraints (i.e., $Gx(t)> g$) can occur in closed-loop, as the uncertainties come into effect. The adaptive constraint tightening in~\eqref{chance_mpc_previous_2} attempts to reduce the conservatism inherent to~\eqref{chance_mpc_previous_1} by incorporating past state constraint violation behavior of the system in closed-loop, but can still be over-conservative (see Remark~\ref{Approx_Chance_Real}).

\begin{remark}[\textbf{Over-conservativeness of previous approaches}] \label{Approx_Chance_Real}
\eqref{chance_mpc_previous_1} approximates~\eqref{chance_real} conservatively 
\cite{munoz2018stochastic},~\cite[Sec. II-A]{korda2014stochastic}, as~\eqref{chance_mpc_previous_1} requires the constraint satisfaction conditionally on $x(t)$ (i.e., for $x(t)$ that can be reached at time $t$ by the given control policy under the uncertainty sequence). Equation~\eqref{chance_real}, however, requires constraint satisfaction in a 
more relaxed 
average sense (i.e., over all realizations of the uncertainty sequence up to $t$). \textcolor{black}{Moreover,~\eqref{chance_mpc_previous_1} does not consider the memory of past state constraint violations which is critical in the present time-average re-interpretation of chance constraints.}
Incorporating past constraint violations by using the adaptive tightening in~\eqref{chance_mpc_previous_2} can still be conservative in satisfying~\eqref{chance_real} in closed-loop 
due to the over-estimation of the tightening parameter $\tilde h$~\cite{munoz2018stochastic}. Additionally, in~\eqref{chance_mpc_previous_2}, the nominal MPC solutions never violate the state constraints over the prediction horizon,
as a result of restricting the size of the feasible state set (despite a larger feasible state set being available to the controller as compared to the nominal MPC when accommodating for uncertainty), which the MPC optimizer can theoretically exploit to further reduce conservativeness in closed-loop.
\end{remark}

Based on {Remark}~\ref{Approx_Chance_Real}, 
which shows that both~\eqref{chance_mpc_previous_1} and~\eqref{chance_mpc_previous_2} can be over-conservative in satisfying~\eqref{chance_real} in closed-loop, we propose to adaptively relax the state constraints in the nominal MPC instead of tightening them. The adaptive relaxation allows for state constraint violations over the nominal MPC prediction horizon ($Gx(t\!+\!k|t)>g$ with $\textbf{\text w}(t) = \textbf{0}$), to \textit{push} the system towards reduced conservativeness. We relax the satisfaction of~\eqref{chance_mpc_previous_1} and approximate~\eqref{chance_real} by reformulating the nominal state constraints as, 
\begin{equation} \label{chance_mpc_now}
\begin{aligned}
Gx(t\!+\!k|t) \leq g- h(t),  \qquad \forall k\in \mathbb{N}_1^{N},\forall t, \textbf{\text w}(t) = \textbf{0},
\end{aligned}
\end{equation}
where $\ h \in \mathbb{R}^{r}$ is the adaptive relaxing parameter with $h_i<0$, $\forall i \in  \mathbb{N}_{1}^{r}$. It should be noted that the sign of $h_i$ in~\eqref{chance_mpc_now} is opposite to $\tilde h_i$ in~\eqref{chance_mpc_previous_2}. We also observe that decreasing $h(t)$ in~\eqref{chance_mpc_now} expands the feasible state set, pushing the system more towards state constraint violations (i.e., $Gx(t\!+\!k|t)> g$), while increasing $h(t)$ contracts the feasible state set pulling the system away from state constraint violations.\footnote{Note that in economic MG dispatch with BESS in VRE grids, where the objective function is the actual economic cost of system operation like the electricity bill, and not necessarily only a penalty on the control input (BESS dispatch), relaxing the state constraints in the nominal MPC does not automatically lead the system to predicted nominal solutions that violate the (original) state constraints pathologically over the nominal MPC prediction horizon. The state (BESS SOC), in these applications tries to exploit the full feasible state set to best reduce economic cost for the MPC prediction horizon. 
Nevertheless, the case where the proposed formulation can result in pathological constraint violations is averted in closed-loop by a post-processing step described later in Section~\ref{post-process} and Remark~\ref{practical}.}
The initial value of $h$ at $t=0$ can be calculated based on domain knowledge~\cite{ghosh2023adaptive}, which obviates the requirement of past uncertainty samples, 
as in~\cite{IEEE_1}. The initial value of $h$ is not important since $h$ gets adapted as the system evolves with time~\cite{IEEE_2}. The ordered collection of adaptive relaxation parameters along the MPC prediction horizon is denoted as,
\begin{equation}
\begin{aligned}
\textbf{\text h}(t) := \Bigl(h(t),h(t),\ldots,h(t)\Bigr)^\top \in \mathbb{R}^{Nr}. \nonumber \\
\end{aligned}
\end{equation}

The nominal OA-SMPC, which is assumed to be a convex optimization problem is then formulated as, 
\begin{subequations}\label{mpc_general}
\begin{equation} \label{objective_mpc}
\begin{aligned}
\textbf{\text u}^*(t) = \mathop{\arg \min}\limits_{\textbf{\text u}(t) \in \mathbb{R}^{Nm}}\: {J}(x(t),\textbf{\text u}(t),\textbf{\text w}(t) = \textbf{0}),
\end{aligned}
\end{equation}
\vspace{-1.5em}
subject to
\begin{equation} \label{dynamics_compact_nominal}
\begin{aligned}
\textbf{\text x}(t+1) = &\textbf{ \text A}x(t)+\textbf{\text B}\textbf{\text u}(t), 
\end{aligned}
\end{equation}
\vspace{-1.5em}
\begin{equation} \label{input_compact_nominal}
\begin{aligned}
\textbf{\text Su}(t) \leq \textbf{\text s},
\end{aligned}
\end{equation}
\vspace{-1.5em}
\begin{equation} \label{equality_compact_nominal}
\begin{aligned}
\textbf{\text Mu}(t) = \textcolor{black}{\textbf{\text c}(t)},
\end{aligned}
\end{equation}
\vspace{-1.5em}
\begin{equation} \label{chance_constraints_compact_nominal}
\begin{aligned}
\textbf{\text Gx}(t+1) \leq \textbf{\text g}-\textbf{\text h}(t).
\end{aligned}
\end{equation}
\end{subequations}

\noindent where $J: \mathbb{R}^{n} \times \mathbb{R}^{Nm} \times \mathbb{R}^{Np} \rightarrow{\mathbb{R}}$ is an arbitrary convex function, $\textbf{\text S} \in \mathbb{R}^{Nq \times Nm}$, $\textbf{\text s} \in \mathbb{R}^{Nq}$, $\textbf{\text M} \in \mathbb{R}^{Nd \times Nm}$, $\textcolor{black}{\textbf{\text c}} \in \mathbb{R}^{Nd}$, $\textbf{\text G} \in \mathbb{R}^{Nr \times Nn}$, $\textbf{\text g} \in \mathbb{R}^{Nr}$. Note that dropping \eqref{equality_compact_nominal} makes the problem setup similar to \cite{IEEE_1}. Note that in~\eqref{chance_constraints_compact_nominal}, the MPC state constraints are applied for $x(t+k|t), \; \forall k \in \mathbb{N}_{1}^{N}$, and not for the present state corresponding to $k\!=\!0$ to allow for the present state to be outside of the feasible state set of the nominal OA-SMPC. \textcolor{black}{Assumption~\ref{assumption_control_inputs}, discussed next, ensures recursive feasibility and existence of an ideal control policy (described later in Assumption~\ref{assumption_ideal_control_policy}).} 

\begin{assumption}[\textbf{Control inputs 
\cite{IEEE_1}}]\label{assumption_control_inputs}
(a) \textcolor{black}{The control input constraints~\eqref{input_compact_nominal} are such that the system can provide enough control input to bring the predicted state at the next time-step, from any present state $x(t)$, to the feasible region of the nominal OA-SMPC~\eqref{chance_constraints_compact_nominal}
. Specifically, 
\begin{align*}
    \exists \; u(t|t) \text{ s.t } Su(t|t) &\leq s, \; Mu(t|t) = c(t|t), \\
    G(x(t+1|t)) &\leq g-h(t),
\end{align*} where $x(t+1|t) = Ax(t)+Bu(t|t)$, for all $t$. }(b) The system is one step controllable.
\end{assumption}

\textcolor{black}{The condition for testing Assumption \ref{assumption_control_inputs}(a) which ensures recursive feasibility (similar to~\cite{IEEE_1,schildbach2012randomized}) of the nominal OA-SMPC is given in the \ref{appendix_main}ppendix, and is excluded here for brevity. Assumption \ref{assumption_control_inputs}(a) ensures that after handling the uncertainty from the previous time step in closed-loop, resulting in the present state $x(t)$, which may be outside the feasible state set of the nominal OA-SMPC~\eqref{chance_constraints_compact_nominal}, the computed control input is strong enough to bring the predicted system state at the next time-step back to the feasible state set. }

\textcolor{black}{Assumption \ref{assumption_control_inputs}(a), while theoretically can be, is generally not restrictive for practical applications such as microgrids or HVAC systems, as these systems are generally designed to be able to have enough control input power to be able to handle uncertainties~\cite{IEEE_1}. Assumption \ref{assumption_control_inputs}(b) is more restrictive and is only used for ensuring sufficient conditions for the existence of an ideal control policy at every time step (see Assumption~\ref{assumption_ideal_control_policy}). Assumption \ref{assumption_control_inputs}(b) can be relaxed for practical applications such as the one described in the case study in Section \ref{case_study}.}


\begin{remark}[\textbf{Structure of the input matrix}]\label{Recursive} 
Note that Assumption~\ref{assumption_control_inputs}(b) is sufficient for saying that the system has at least as many control inputs as states (i.e., $n\leq m$) and $B$ has full row rank. The assumption 
implies that if $n=m$, $B$ has an inverse, while if $n<m$, $B$ has a right inverse.
\end{remark}

\begin{assumption}[\textbf{Form of the $h$ update rule}]\label{assumption_h_update}
The online h (adaptive relaxing parameter) update rule can be written as $h_i(t):=h_i(t-1)[1+K_i(t)]$, where $K_i(t)\!>\!-1, \; \forall t$ ensures $h_i(t)\!<\!0, \; \forall t$~\cite[Eq. (12)]{ghosh2023adaptive}. 
\end{assumption}

\begin{remark}[\textbf{Behavior of the $h$ update rule}]\label{k_t variation}
In {Assumption}~\ref{assumption_h_update}, $K_i(t)>0$ decreases $h_i(t)$, expanding the state constraints, pushing the system more towards state constraint violations, while $K_i(t)<0$ increases $h_i(t)$, contracting the state constraints, pulling the system away from state constraint violations.
\end{remark}

\subsection{Observed Violations} \label{violation_observed}

In this section, for consistency with earlier works like \cite{IEEE_1}, we drop~\eqref{equal_real} and~\eqref{equality_compact_nominal}. Thus, the nominal OA-SMPC computed optimal control inputs for the first time-step of the prediction horizon are implemented in closed-loop. The observed states get corrected once the uncertainties are realized, by using~\eqref{dynamics_real}. The case where~\eqref{equal_real} and~\eqref{equality_compact_nominal} are considered in the problem formulation is discussed in Section~\ref{post-process} which additionally post-processes the nominal OA-SMPC computed control inputs to correct for the uncertainty.

Without loss of generality, consider the $i^{\text{th}}$ state constraint in~\eqref{chance_real}. Let $V_i(t+1)\in \{0,1\}$ track whether the state constraint is violated in closed-loop at time $t+1$, while $Y_i(t+1) \in [0,1]$ keeps track of the time-average of violations up to time $t+1$. 
Note that the control input applied at time $t$ (along with the uncertainty realized at $t$) is manifested with updated system states, which can be observed only at $t+1$, i.e., there is 1 time-step delay in observing violations (or non-violations) from the time when the control inputs and uncertainties are applied.
\begin{subequations}\label{violations}
\begin{equation} \label{V_general}
\begin{aligned}
\!\!\!\!V_i(t+1)\!:= \!\left\{ 
\begin{matrix}
1,  &\!\!\!G_i(Ax(t)+Bu^*(t|t)+ Ew(t))>g_i,\\
0,  &\!\!\!G_i(Ax(t)+Bu^*(t|t)+ Ew(t))\leq g_i.
\end{matrix}
\right.
\end{aligned}
\end{equation} 
\begin{equation} \label{Y_general}
\begin{aligned}
Y_i(t+1) := \frac{\sum_{j=1}^{t+1}V_{i}(j)}{t+1}.
\end{aligned}
\end{equation}
\end{subequations}

The framework for tracking the state constraint violations and time-average of violations in the case of JCC, is the same as that of the individual chance constraints described in~\eqref{violations}. The only difference in the case of JCC is that a violation occurs if any one of the constraints (involved in the JCC) violates its specific state constraint bounds in closed-loop. 
\vspace{-0.5em}

\subsection{Convergence Properties of \texorpdfstring{$Y(t)$}{Yt}}\label{optimal_h}
In this section, like~\ref{violation_observed}, without loss of generality, we limit our discussion to the $i^{\text {th}}$ state constraint in~\eqref{chance_real} with $ i \in  \mathbb{N}_{1}^{r}$, with its corresponding adaptive relaxation parameter $h_i \in \mathbb{R}_{0-}$. The conditions established for ensuring the 
convergence of the time-average of state constraint violations to the maximum allowable violation probability in this paper uses a similar simplifying assumption as in~\cite{IEEE_1}. 
The 
simplifying assumption (see Assumption~\ref{assumption_ideal_control_policy}) allows the controller to apply ideal control inputs at the current time-step leading to a desired probability of violation of state constraints at the next time-step, under unknown bounded uncertainties.

Denote by $Z_i(t)=|\alpha_i - Y_i(t)|$ the absolute difference between the maximum allowable violation probability and time-average of violations of the $i^{\text {th}}$ state constraint observed at time $t$. We have to ensure that $Y_i$ tends to $\alpha_i$ in closed-loop as the system evolves with time for non-conservative chance constraint satisfaction. The non-conservative strategy ideally leads to lower costs without violating the state constraints beyond the maximum allowable violation probability. 

\textcolor{black}{
As $Z_i(t)$ is non-negative, the convergence of $Z_i$ can be guaranteed a.s., if $Z_i(t)$ is a supermartingale~\cite{williams1991probability}. 
Following Section~\ref{definitions}, the three conditions for $Z_i$ being a supermartingale are investigated below:}
\begin{enumerate}[label=(\alph*)]
    \item \textcolor{black}{Let $\mathcal{F}_{t}= \sigma{(\{w(0),w(1),\ldots,w(t-1)\})}$ be a $\sigma$-algebra on uncertainties realized up to time $t-1$. $Z_i(t)$ depends on $Y_i(t)$, which depends on the realization of all the uncertainties up to time $t\!-\!1$. Since $Z_i(t)$ is exactly known with information available up to time $t\!-\!1$, $Z_i(t)$ is $\mathcal{F}_{t}$ measurable $\forall t> 0$, and is thus adapted. Also, since $Z_i(t+1)$ is random with information available in $\mathcal{F}_{t}$, the process $Z_i$ is stochastic.
    \item As $V_i(t) \in \{0,1\}$, $Y_i(t) \in [0,1]$, and $\alpha_i \in (0,0.5)$, thus $Z_i(t)=|\alpha_i - Y_i(t)| \in [0,1).$ Thus, $\mathbb{E}[|Z(t)|] < 1< \infty$, \; $\forall t > 0$.
    \item It remains to show $\mathbb{E}[Z_i(t+1)|\mathcal{F}_{t}] \leq Z_i(t)$ a.s., $\forall t>0$, which we will show to hold under similar simplifying assumptions as in~\cite{IEEE_1}. The assumption involves replacing the stochastic $Z_i(t+1)|\mathcal{F}_t$ by its ideal surrogate $Z_i^*(t+1)|\mathcal{F}_t$ as explained next.}
\end{enumerate}

\noindent \textcolor{black}{From~\eqref{Y_general} $Y_i(t+1)$ can be written as,
\begin{equation} \label{Y(t+1)}
\begin{aligned}
Y_i (t+1) = \sum_{j=1}^{t+1}\frac{V_i(j)}{t+1} = t\frac{Y_i(t)}{t+1} +\frac{V_i(t+1)}{t+1}.
\end{aligned}
\end{equation}
Denoting $\mathbb{E}[Z_i(t\!+\!1)|\mathcal{F}_{t}]-Z_i(t)$ by $\Delta_i(t)$ and substituting~\eqref{Y(t+1)} in~$\Delta_i(t)$ results in,
\begin{align}
\Delta_i(t) = &\mathbb{E}\Biggl[\left|\alpha_i - t\frac{Y_i(t)}{t+1} -\frac{V_i(t+1)}{t+1}\right|| \mathcal{F}_t\Biggr] \nonumber \\
-& |\alpha_i - Y_i(t)|.\label{Delta}
\end{align}
Let $p_i(t+1) := \mathbb{P}(V_i(t\!+\!1) = 1|\mathcal{F}_{t})$, which means that ${p_i(t+1)}$ is the probability of observing a constraint violation at time $t\!+\!1$. Similarly, $1-p_i(t+1) = \mathbb{P}(V_i(t\!+\!1) = 0|\mathcal{F}_{t})$ is the probability of not observing a violation at time $t\!+\!1$.   
We notice that the only stochastic part within the expectation in the RHS of~\eqref{Delta} is $V_i(t+1)$ based on $\mathcal{F}_t$. Therefore,~\eqref{Delta} can be rewritten in terms of $p_i(t+1)$ to replace the expectation as,
\vspace{-1em}
\begin{equation} \label{Delta_2}
\begin{aligned}
\Delta_i(t) = &p_i(t+1)\Biggl[\left|\alpha_i - t\frac{Y_i(t)}{t+1} -\frac{1}{t+1}\right|\Biggr] \\
+(1-&p_i(t+1))\Biggl[\left|\alpha_i - t\frac{Y_i(t)}{t+1}\right|\Biggr] -
|\alpha_i - Y_i(t)|. 
\end{aligned}
\end{equation}
Simplifying yields,
\begin{align}
&\Delta_i(t) = p_i(t+1)\beta_i(t) + \left[{\left|\alpha_i - t\frac{Y_i(t)}{t+1}\right| -
|\alpha_i - Y_i(t)|}\right], \label{Delta_3}\\
&\beta_i(t)=\left|\alpha_i - t\frac{Y_i(t)}{t+1} -\frac{1}{t+1}\right|-\left|\alpha_i - t\frac{Y_i(t)}{t+1}\right|.\label{Delta_4}
\end{align}}

\textcolor{black}{To keep the analysis applicable to any arbitrary probability distribution of $w(t)$, we consider the sign of $\beta_i(t)$ to decide the ideal control policy~\cite{IEEE_1}, for which we introduce Assumption~\ref{assumption_ideal_control_policy}.} 

\begin{assumption}[\textbf{\textcolor{black}{Ideal control policy}}]\label{assumption_ideal_control_policy}
\textcolor{black}{(a) There exists an ideal control policy (or ideal control input) at time $t$, applying which makes $p_i(t+1)=p^*_i(t+1)$, where $p^*_i(t+1)\in\{0,1\}, \; \forall t$. (b) When $ \beta_i(t)< 0$, we apply ideal control inputs at $t$ leading to $p_i^*(t+1)=1$ whereas, if $\beta_i(t) > 0$, we apply ideal control inputs at $t$ leading to $p_i^*(t+1)=0$.}
\end{assumption}

\begin{definition}[\textbf{\textcolor{black}{Ideal surrogate of a variable}}]\label{ideal_surrogate}
\textcolor{black}{The manifestation of a variable under the ideal control policy in Assumption~\ref{assumption_ideal_control_policy} is defined as the ideal surrogate of that variable. It is denoted by an asterisk after the variable.}
\end{definition}

\textcolor{black}{Assumption~\ref{assumption_ideal_control_policy}(a) ensures that there exists 
a control input which causes a constraint violation a.s. at the next time step $t+1$ in closed-loop from any present state $x(t)$, i.e., $p^*_i(t+1)=1$. Similarly, $p^*_i(t+1)=0$ means that the controller can drive the system to prevent a constraint violation a.s. at $t+1$ in closed-loop.}
\textcolor{black}{$Z_i(t)$ and $\Delta_i(t)$ under the ideal control policy are referred to as $Z_i^*(t)$ and $\Delta_i^*(t)$ respectively,\footnote{\textcolor{black}{$Z_i^*(t)=|\alpha_i-Y_i^{*}(t)|$ and $\Delta_i^*(t)=Z_i^*(t+1)-Z_i(t)$.}} consistent with Definition~\ref{ideal_surrogate}. Assumption~\ref{assumption_ideal_control_policy}(b) ensures that when $ \beta_i(t)< 0$, we choose ideal control inputs leading to $p_i^*(t+1)=1$ so that $\Delta_i^*(t)$ may be $\leq 0$ a.s. Similarly, when $\beta_i(t) > 0$, Assumption~\ref{assumption_ideal_control_policy}(b) ensures that we choose ideal control inputs leading to $p_i^*(t+1)=0$ so that $\Delta_i^*(t)$ may be $\leq 0$ a.s. Note that despite the application of ideal control inputs, the second and third term in the RHS of ~\eqref{Delta_3} can lead to  $\Delta_i^*(t)>0$, which is used to derive the critical region $\kappa(\alpha_i,t)$, in Theorem~\ref{monotonic_conv} later. The critical region signifies a region where, if $Y_i(t) \in \kappa(\alpha_i,t)$, then $\Delta_i^*(t)>0$ a.s. }
 
\textcolor{black}{The case when $\beta_i(t)=0$ implies 
$p_i^*(t+1)$ not having an effect on $\Delta_i(t)$, as the first term in the RHS of~\eqref{Delta_3} vanishes regardless of the value of $p_i(t+1)$. From~\eqref{Delta_4}, it can be shown that $\beta_i(t)=0\Leftrightarrow \; Y_i(t)=\frac{\alpha_i -\frac{1}{2(t+1)}}{1-\frac{1}{(t+1)}}$, and further solving for $\Delta_i(t) >0$ in~\eqref{Delta_3}, yields $\alpha_i>\frac{1}{2(t+1)}$, which becomes more likely to be satisfied as $t$ increases.
However, the case $\beta_i(t) = 0$ can be avoided by a particular choice of $\alpha_i$. 
From~\eqref{Delta_4}, 
\begin{align*}
    \beta_i(t) \neq 0 \Leftrightarrow  Y_i(t) \neq \frac{\alpha_i -\frac{1}{2(t+1)}}{1-\frac{1}{(t+1)}}.
\end{align*}
Since, $Y_i(t)=\frac{\sum_{j=1}^{t}V_{i}(j)}{t}$, and $\sum_{j=1}^{t}V_{i}(j) \in \mathbb{N}$, therefore, 
\begin{align*}
    Y_i(t) \neq \frac{\alpha_i -\frac{1}{2(t+1)}}{1-\frac{1}{(t+1)}} \iff \sum_{j=1}^{t}V_{i}(j) \neq (t+1)\alpha_i-\frac{1}{2},
\end{align*}
which can be ensured by appropriate choice of $\alpha_i$. Specifically, $(t+1)\alpha_i-\frac{1}{2} \not \in \mathbb{N}, \; \forall t$. The online adaptive relaxation rule is introduced next in Section~\ref{h_section}, with its behavior under practical scenarios being discussed in Section~\ref{deviation}.}

\begin{subsubsection}{\textcolor{black}{\texorpdfstring{$h$}{h} update rule}}\label{h_section}
\textcolor{black}{$\beta_i(t) < 0$ implies $\alpha_i-Y_i(t)+\frac{2Y_i(t)-1}{2(t+1)} > 0$,\footnote{\textcolor{black}{$\beta_i(t)=|\zeta_i(t)-\frac{1}{t+1}|-|\zeta_i(t)|$, where $\zeta_i(t)=\alpha_i-t\frac{Y_i(t)}{t+1}$. Solving for $\beta_i(t)\!<\!0$ leads to $\zeta_i(t)>\!\frac{1}{2(t+1)}$, which implies $\alpha_i\!-\!Y_i(t)\!+\!\frac{2Y_i(t)-1}{2(t+1)}\! >\! 0$.}} and is associated with {violation} of state constraints at $t\!+\!1$ which is achieved by expanding the state constraint limits in~\eqref{chance_constraints_compact_nominal}. Similarly $\beta_i(t) > 0$ implies $\alpha_i-Y_i(t)+\frac{2Y_i(t)-1}{2(t+1)} < 0$ and is associated with {non-violation} of state constraints at $t\!+\!1$ which is achieved by contracting the limits in~\eqref{chance_constraints_compact_nominal}.
From {Assumption}~\ref{assumption_h_update}, the $h_i$ update rule can be framed with $K_i(t)\propto \bigl[\alpha_i-Y_i(t)+\frac{2Y_i(t)-1}{2(t+1)}\bigr]$ as 
\vspace{-0.5em}
\begin{equation}\label{h_update_2}
\begin{aligned}
h_i(t)=h_i(t-1)\biggl[1+\frac{\alpha_i-Y_i(t)+\frac{2Y_i(t)-1}{2(t+1)}}{\gamma_i}\biggr], \end{aligned}
\end{equation}
\textcolor{black}{where $\gamma_i \in \mathbb{R}_{0+}$ is a constant of proportionality that adjusts the rate of $h_i$ update ensuring $\frac{\alpha_i-Y_i(t)+\frac{2Y_i(t)-1}{2(t+1)}}{\gamma_i}>-1, \; \forall t$}.}
%
%
\textcolor{black}{Note that Theorem~\ref{monotonic_conv}, discussed next, implicitly assumes that~\eqref{h_update_2} is able to enforce Assumption~\ref{assumption_ideal_control_policy}. However, it may be possible to devise other $h$ update rules enforcing Assumption~\ref{assumption_ideal_control_policy}, wherein Assumption~\ref{assumption_h_update} (and consequently, use of~\eqref{h_update_2}) can be relaxed. In practical applications, while~\eqref{h_update_2} cannot guarantee satisfaction of Assumption~\ref{assumption_ideal_control_policy} in closed-loop, it still encourages it (see Section~\ref{deviation}). }

\end{subsubsection}

\begin{theorem}\label{monotonic_conv}
\textcolor{black}{Let Assumptions~\ref{assumption_system},~\ref{assumption_uncertainty},~\ref{assumption_control_inputs}, and~\ref{assumption_ideal_control_policy} hold. Given $\alpha_i\!>\!\frac{1}{2(t_0+1)}$,
$Z_i^*(t)$, which is the ideal surrogate of the desired supermartingale $Z_i(t)$, is monotonically decreasing a.s. $\forall t \geq t_0$, if and only if, ${Y_i(t) \not\in \kappa (\alpha_i ,t)}, \; \forall t \geq t_0$, where $\kappa (\alpha_i ,t)$ is a neighborhood of $\alpha_i$ 
defined as, 
\begin{equation} \label{Kappa}
\begin{aligned}
\kappa (\alpha_i ,t) := \Biggl(\frac{\alpha_i-\frac{1}{2(t+1)}}{1-\frac{1}{2(t+1)}},\frac{\alpha_i}{1-\frac{1}{2(t+1)}}\Biggr). \; 
\nonumber \\
\end{aligned}
\end{equation}} 
\end{theorem}
\begin{proof}\label{proof_1}
\vspace{-1em}
Theorem~\ref{monotonic_conv} says that given $\alpha_i>\frac{1}{2(t_0+1)}$, $\mathcal{P} \iff \mathcal{Q}$, where $\mathcal{P}$ is defined as $\Delta_i^*(t)\leq0$ a.s., $\forall t \geq t_0$, and $\mathcal{Q}$ is defined as ${Y_i(t) \not\in \kappa (\alpha_i ,t)}, \; \forall t \geq t_0$.
%
To prove, $\mathcal{P} \iff \mathcal{Q}$, we will first prove the inverse as true, i.e., $\neg \mathcal{P} \implies \neg \mathcal{Q}$, which implies $\mathcal{Q} \implies \mathcal{P}$. Then we will prove the contrapositive as true, i.e., $\neg \mathcal{Q} \implies \neg \mathcal{P}$, which implies $\mathcal{P} \implies \mathcal{Q}$, which will complete the proof.

To prove $\neg \mathcal{P} \implies \neg \mathcal{Q}$, assume $\neg \mathcal{P}$ is true, i.e., $\Delta_i^*(t) >0$ a.s. To show, $\neg \mathcal{Q}$ holds, we consider Cases~\ref{case_1},~\ref{case_2} and~\ref{case_3} based on the sign of $\beta_i(t)$ 
below.

\begin{case}{1}\label{case_1}
\begin{subequations}\label{contradiction_1}
\begin{equation}\label{contradiction_1a}
\begin{aligned}
&\beta_i(t) < 0 \Leftrightarrow
\;{p^*_i (t+1) = 1}
\;  \\ &\Leftrightarrow \; {\alpha_i - Y_i(t) + \frac{2Y_i(t)-1}{2(t+1)} > 0} 
\Leftrightarrow Y_i(t)<\frac{\alpha_i -\frac{1}{2(t+1)}}{1-\frac{1}{t+1}}. 
\end{aligned}
\end{equation} 
\emph{Solving for $\Delta_i^*(t) >0$ when ${p^*_i (t+1) = 1}$ yields}
\begin{equation}\label{contradiction_1b}
\begin{aligned}
Y_i(t)>\frac{\alpha_i -\frac{1}{2(t+1)}}{1-\frac{1}{2(t+1)}}.
\end{aligned}
\end{equation}
\end{subequations}
\noindent \emph{Equation~\eqref{contradiction_1} implies the critical region of Case~\ref{case_1} is $\kappa_{1}(\alpha_i ,t) = \Bigl(\frac{\alpha_i -\frac{1}{2(t+1)}}{1-\frac{1}{2(t+1)}},\frac{\alpha_i -\frac{1}{2(t+1)}}{1-\frac{1}{t+1}}\Bigr)$, implying if $Y_i(t) \in \kappa_{1}(\alpha_i ,t)$, despite applying ideal control inputs to the system at $t$ leading to ${p^*_i (t+1) = 1}$, $\Delta_i^*(t) >0$ a.s.}
\end{case}

\begin{case}{2}\label{case_2}
\begin{subequations}\label{contradiction_2}
\begin{equation}\label{contradiction_2a}
\begin{aligned}
& \beta_i(t) > 0 \Leftrightarrow \; {p^*_i (t+1) = 0}  \\ &\Leftrightarrow \;{\alpha_i - Y_i(t) + \frac{2Y_i(t)-1}{2(t+1)} < 0}
\Leftrightarrow Y_i(t)>\frac{\alpha_i -\frac{1}{2(t+1)}}{1-\frac{1}{t+1}}. 
\end{aligned}
\end{equation} 
\emph{Solving for $\Delta_i^*(t) >0$ when ${p^*_i (t+1) = 0}$ yields}
\begin{equation}\label{contradiction_2b} 
\begin{aligned}
Y_i(t)<\frac{\alpha_i}{1-\frac{1}{2(t+1)}}.
\end{aligned}
\end{equation}
\end{subequations}

\noindent \emph{Equation~\eqref{contradiction_2} implies the critical region of Case~\ref{case_2} is $\kappa_{2}(\alpha_i ,t) = 
 \left(\frac{\alpha_i -\frac{1}{2(t+1)}}{1-\frac{1}{t+1}},\frac{\alpha_i}{1-\frac{1}{2(t+1)}}\right)$, implying if $Y_i(t) \in \kappa_{2}(\alpha_i ,t)$, despite applying ideal control inputs to the system at $t$ leading to ${p^*_i (t+1) = 0}$, $\Delta_i^*(t) >0$  a.s.}
\end{case}

\begin{case}{3}\label{case_3}
\emph{$\beta_i(t)=0$ leads to $\Delta_i^*(t) > 0$ a.s. because $\alpha_i>\frac{1}{2(t_0+1)} \geq \frac{1}{2(t+1)}, \; \forall t\geq t_0$. Additionally, $\beta_i(t)=0 \Leftrightarrow Y_i(t) = \frac{\alpha_i -\frac{1}{2(t+1)}}{1-\frac{1}{(t+1)}}=\kappa_{3}(\alpha_i ,t)$. }
\end{case}
\noindent The total critical region from Cases~\ref{case_1},~\ref{case_2} and~\ref{case_3} yields, 
\begin{align*}
    \kappa (\alpha_i ,t) =& \kappa_{1} (\alpha_i ,t)\; \cup \; \kappa_{2} (\alpha_i ,t) \; \cup \; \kappa_{3} (\alpha_i ,t)\ \\
    =& \Biggl(\frac{\alpha_i-\frac{1}{2(t+1)}}{1-\frac{1}{2(t+1)}},\frac{\alpha_i}{1-\frac{1}{2(t+1)}}\Biggr),
\end{align*}

\noindent which shows that if 
$\Delta_i^*(t) >0$  a.s., 
then $Y_i(t) \in \kappa (\alpha_i ,t)$, which proves  $\neg \mathcal{P} \implies \neg \mathcal{Q}$.

To prove $\neg \mathcal{Q} \implies \neg \mathcal{P}$, assume, $\neg \mathcal{Q}$, i.e, $Y_i(t) \in \kappa (\alpha_i ,t)$. $\neg \mathcal{Q}$ is subdivided into Cases~\ref{case_4},~\ref{case_5} and~\ref{case_6}.

\begin{case}{4}\label{case_4} \emph{$\beta_i(t)<0 \Leftrightarrow {p^*_i (t+1) = 1}$, which yields~\eqref{contradiction_1a}. 
\eqref{contradiction_1a} and $Y_i(t) \in \kappa (\alpha_i ,t)$ yields, $Y_i(t) \in \Bigl(\frac{\alpha_i -\frac{1}{2(t+1)}}{1-\frac{1}{2(t+1)}},\frac{\alpha_i -\frac{1}{2(t+1)}}{1-\frac{1}{t+1}}\Bigr) = \kappa_{1}(\alpha_i ,t)$. 
\noindent However, from Case~\ref{case_1}, we know that when $p^*_i (t+1) = 1$ and $Y_i(t) \in \kappa_{1}(\alpha_i ,t)$, then $\Delta_i^*(t) >0$  a.s.}

\end{case}

\begin{case}{5}\label{case_5}
\emph{
$\beta_i(t) >0 \Leftrightarrow p^*_i (t+1) = 0$,  which yields~\eqref{contradiction_2a}. 
\eqref{contradiction_2a} and $Y_i(t) \in \kappa (\alpha_i ,t)$ yields, $Y_i(t) \in \left(\frac{\alpha_i -\frac{1}{2(t+1)}}{1-\frac{1}{t+1}},\frac{\alpha_i}{1-\frac{1}{2(t+1)}}\right) = \kappa_{2}(\alpha_i ,t)$. \noindent However, from Case~\ref{case_2}, we know that when $p^*_i (t+1) = 0$ and $Y_i(t) \in \kappa_{2}(\alpha_i ,t)$, then $\Delta_i^*(t) >0$  a.s.}

\end{case}

\begin{case}{6}\label{case_6} \emph{$\beta_i(t) = 0 \Leftrightarrow Y_i(t) = \kappa_3(\alpha_i,t)$ which leads to $\Delta_i^{*}(t) >0$ a.s. because $\alpha_i>\frac{1}{2(t+1)}$.}

\end{case}
Cases~\ref{case_4},~\ref{case_5} and~\ref{case_6} together prove  $\neg \mathcal{Q} \implies \neg \mathcal{P}$, which completes the proof.
\end{proof}

\begin{remark}[\textbf{\textcolor{black}{Extension} of Theorem~\ref{monotonic_conv}}]\label{implication}

Theorem~\ref{monotonic_conv} establishes the conditions necessary and sufficient for $\Delta_i^*(t)\leq0$ a.s., $\forall t \geq t_0$, which implies that $Z_i^*(t)$ is monotonically decreasing a.s., $\forall t \geq t_0$. \textcolor{black}{From Cases \ref{case_1}, \ref{case_4}, and Cases \ref{case_2}, \ref{case_5}, if the critical region is defined by $\kappa' (\alpha_i ,t) := \Biggl[\frac{\alpha_i-\frac{1}{2(t+1)}}{1-\frac{1}{2(t+1)}},\frac{\alpha_i}{1-\frac{1}{2(t+1)}}\Biggr]$, then $\Delta_i^*(t)<0$ a.s., $\forall t \geq t_0$, which implies that $Z_i^*(t)$ is strictly decreasing a.s., $\forall t \geq t_0$. As $Z_i^*(t)$ is bounded from below by $0$, 
we conclude $\lim_{t\to\infty}  Z_i^*(t)=\inf \{Z_i^*(t)\}=0$} a.s. from the monotone convergence theorem, which implies 
the \textcolor{black}{asymptotic} convergence of $Y_i^{*}$ to $\alpha_i$ a.s., if $\alpha_i>\frac{1}{2(t_0+1)}$, $p_i(t+1)=p_i^*(t+1)$, and \textcolor{black}{${Y_i(t) \not\in \kappa' (\alpha_i ,t)}, \;\forall t \geq t_0$.\footnote{\textcolor{black}{Note that although $Y_i(t_0)$ need not necessarily be $Y_i^*(t_0)$, we abuse the notation to use $Z_i(t_0)$ and $Z_i^*(t_0)$ interchangeably in Theorem~\ref{monotonic_conv} and Remark~\ref{implication} to aid in a more intuitive explanation.} }
}
\end{remark}

\begin{remark}[\textbf{Width of the critical region~\cite{IEEE_1}}]\label{width}

The width of the critical region \textcolor{black}{${\kappa' (\alpha_i ,t)}$} is $\frac{1}{2t+1}$, while $Y_i(t)\in\{0, \frac{1}{t}, \frac{2}{t}, \ldots, 1\}$. The granularity in possible values of $Y_i(t)$ is $1/t$, $\forall t$. As $\frac{1}{2t+1}<\frac{1}{t}$, there is at most one critical value of $Y_i(t)$, such that despite adapting the system to take ideal control inputs leading to $p_i(t+1)=p_i^*(t+1)$, it is possible that $\Delta_i^*(t)\geq0$ which can lead to the deviation of $Y_i^{*}$ from $\alpha_i$ momentarily. However, the width of the critical region monotonically decreases as $t$ increases, therefore, the assumption of ${Y_i(t) \not\in \textcolor{black}{\kappa' (\alpha_i ,t)}}$, is weak and can be easily satisfied as $t$ increases.  The width of the critical region in the present work is also smaller than~\cite{IEEE_1}, implying more relaxed conditions for guaranteeing convergence of $Y_i^{*}$ to $\alpha_i$.
\end{remark}

\begin{lemma}\label{lemma_1}
\textcolor{black}{Let Assumptions~\ref{assumption_system},~\ref{assumption_uncertainty},~\ref{assumption_control_inputs}, and~\ref{assumption_ideal_control_policy} hold.} Given any initial time $t_0$, with $\alpha_i>\frac{1}{2(t_0+1)}$ and ${Y_i(t_0) \not\in \textcolor{black}{\kappa' (\alpha_i ,t_0)}}$, 
there exists some time $t'=\inf\{t > t_0 \:|\:{Y_i^{*}(t) \in \textcolor{black}{\kappa' (\alpha_i ,t)}\}}$ a.s.

\begin{proof}\label{proof_cor_1}
{The lemma states that if the initial 
time-average of violations of state constraints is outside of the critical region, 
and we apply ideal control inputs from thereon, then at some future time $t'$, the time-average of violations comes inside the critical region a.s. Lemma~\ref{lemma_1} implies the a.s. \textcolor{black}{strict} decrease of $Z_i^*(t)$ is violated for $t\geq t'$ (see Remark~\ref{implication}). 
Lemma~\ref{lemma_1} can be proved by showing that when $Z_i^*(t)$ decreases, then $\frac{\delta{\hat{w_i}(t)}}{\delta{t}} > \frac{\delta{Z^*_i(t)}}{\delta{t}}$, where $\hat w(t)$ is the width of the critical region at time $t$, and $\delta$ is the forward finite difference operator.

$\hat w_i(t):=\frac{1}{2t+1}$, implying $\delta{\hat{w_i}}(t):=\hat w_i(t+1) - \hat w_i(t) = \frac{-2}{(2t+1)(2t+3)}$ where $\delta{t}:=1$. Thus, $\hat w_i(t)$ decreases in the order of $\mathcal{O}(\frac{1}{t^2})$. 

Similarly, $\delta{Z^*_i(t)}:=Z^*_i(t+1)-Z_i(t)=\Delta_i^*(t)$. $\Delta_i^*(t)$ can be subdivided into Cases~\ref{case_7} and~\ref{case_8}. 

\begin{case}{1}\label{case_7}
\emph{When $p^*_i(t+1)=1$, $\Delta_i^*(t)=\biggl[\left|\alpha_i\! -\! t\frac{Y_i(t)}{t+1} -\frac{1}{t+1}\right|\biggr] -|\alpha_i - Y_i(t)|. $ }   
\end{case}

\begin{case}{2}\label{case_8}
\emph{When $p^*_i(t+1)=0$, $\Delta_i^*(t)=\left|\alpha_i - t\frac{Y_i(t)}{t+1}\right| -
|\alpha_i - Y_i(t)|.$ }   
\end{case}

Substituting $Y_i(t)=\frac{\sum_{j=1}^{t}V_{i}(j)}{t}$, where $\sum_{j=1}^{t}V_{i}(j)\leq t$,
in both Cases~\ref{case_7} and~\ref{case_8}, we see that when $Z_i^*(t)$ decreases a.s. 
by virtue of ${Y_i(t) \not\in \textcolor{black}{\kappa' (\alpha_i ,t)}}$, then $Z_i^*(t)$ decreases most modestly (i.e., with the least magnitude of decrease), in the order of $\mathcal{O}(\frac{1}{t})$, which proves that, $\frac{\delta{\hat{w_i}(t)}}{\delta{t}} >\frac{\delta{Z^*_i(t)}}{\delta{t}}$, which completes the proof.
}
\end{proof}
\end{lemma}

\begin{lemma}\label{lemma_2}
\textcolor{black}{Let Assumptions~\ref{assumption_system},~\ref{assumption_uncertainty},~\ref{assumption_control_inputs}, and~\ref{assumption_ideal_control_policy} hold.} Given any initial time $t_0$, with $\alpha_i>\frac{1}{2(t_0+1)}$ and ${Y_i(t_0) \in \textcolor{black}{\kappa' (\alpha_i ,t_0)}}$, 
there exists some time $t'=\inf\{t > t_0 \:|\:{Y_i^{*}(t) \not\in \textcolor{black}{\kappa' (\alpha_i ,t)}}\}$ a.s.

\begin{proof}\label{proof_cor_2}
{The lemma states that if the initial 
time-average of violation of state constraints is inside the critical region, 
and we apply ideal control inputs from thereon, then at some future time $t'$, the time-average of violation goes outside of the critical region a.s. Lemma~\ref{lemma_2} prevents the monotonic increase of $Z_i^*(t)$ for $t\geq t'$. \textcolor{black}{When ${Y_i(t) \in {\kappa' (\alpha_i ,t)}}$, $\frac{\delta{Z^*_i(t)}}{\delta{t}} \geq 0$, while $\frac{\delta{\hat{w_i}(t)}}{\delta{t}} <0$, thus leading to $\frac{\delta{Z^*_i(t)}}{\delta{t}} > \frac{\delta{\hat{w_i}(t)}}{\delta{t}}$,
which completes the proof.}}
\end{proof}
\end{lemma}

\begin{theorem}\label{just_conv}
\textcolor{black}{Let Assumptions~\ref{assumption_system},~\ref{assumption_uncertainty},~\ref{assumption_control_inputs}, and~\ref{assumption_ideal_control_policy} hold.} Given any initial time $t_0$ with $\alpha_i>\frac{1}{2(t_0+1)}$, and time-average of violations $Y_i(t_0)$, 
$Y_i^{*}$ asymptotically converges to $\alpha_i$ a.s.
\end{theorem}

\begin{proof}
{The theorem states that if ideal control inputs are applied to the system starting from any arbitrary time $t_0$ with $\alpha_i>\frac{1}{2(t_0+1)}$, then the time-average of violations converges to the maximum probability of violations of state constraints a.s. The proof follows from Theorem~\ref{monotonic_conv}, Remarks~\ref{implication} and~\ref{width}, Lemmas~\ref{lemma_1},~\ref{lemma_2}, and relaxes the assumption of ${Y_i(t) \not\in \textcolor{black}{\kappa' (\alpha_i ,t)}},\; \forall t\geq t_0$, of Remark~\ref{implication}. Theorem~\ref{just_conv} also rigorously proves the result which was intuitively argued in~\cite{IEEE_1}. 

From Remark~\ref{width}, the width of the critical region $\hat w_i(t)$ is $\frac{1}{2t+1}$, which monotonically decreases with time. As the critical region $\textcolor{black}{\kappa' (\alpha_i ,t)}$ is a neighborhood of $\alpha_i$ by construction, decrease of $\hat w_i(t)$ implies decrease of the width of the neighborhood. $\lim_{t\to\infty} \hat w_i(t)=0$, which implies $\textcolor{black}{\lim_{t\to\infty} \kappa' (\alpha_i ,t)=\alpha_i}$, i.e., the critical region becomes a point. With $\alpha_i>\frac{1}{2(t_0+1)}$ and $p_i(t+1)=p_i^*(t+1), \; \forall t\geq t_0$, two cases are possible: 
\begin{case}{1}\label{case_9}
\emph{
When ${Y_i(t_0) \not\in \textcolor{black}{\kappa' (\alpha_i ,t)}}$, Lemma~\ref{lemma_1} concludes that there exists some $t'=\inf\{t > t_0 \:|\:{Y_i^{*}(t) \in \textcolor{black}{\kappa' (\alpha_i ,t)}\}}$ a.s, which implies that $Y_i^{*}(t')$ is closer to $\alpha_i$ as compared to $Y_i(t_0)$ a.s. (see Remark~\ref{implication}). Then, from Lemma~\ref{lemma_2}, we can conclude that there exists some $t''=\inf\{t > t' \:|\:{Y_i^{*}(t) \not\in \textcolor{black}{\kappa' (\alpha_i ,t)}}\}$ a.s., which implies that $Y_i^{*}(t'')$ is no closer to $\alpha_i$ as compared to $Y_i^{*}(t)'$ a.s. (see Remark~\ref{implication}). 
The process repeats, with $Y_i^{*}$ coming in and out of the critical region a.s. as time progresses. However, as the width of the critical region 
vanishes at $t\to\infty$, making the critical region the point $\alpha_i$, we conclude $Y_i^{*}$ converges to $\alpha_i$ a.s. as $t\to\infty$.
}
\end{case}
\begin{case}{2}\label{case_10}
\emph{
When ${Y_i(t_0) \in \textcolor{black}{\kappa' (\alpha_i ,t)}}$, a similar argument can be made as in Case~\ref{case_9} involving Lemma~\ref{lemma_1} and~\ref{lemma_2}, with $Y_i^{*}$ coming in and out of the critical region a.s. until ultimately converging to $\alpha_i$ a.s. as $t \to \infty$.   } 
\end{case}
Both Cases~\ref{case_9} and~\ref{case_10} complete the proof. 
}
\end{proof}

\begin{subsubsection}{\textcolor{black}{Deviation of the practical control policy from the ideal control policy}}\label{deviation}
\textcolor{black}{
Remark \ref{width} establishes that the assumption of ${Y_i(t) \not\in \textcolor{black}{\kappa' (\alpha_i ,t)}}$, is weak and can be easily satisfied as $t$ increases. Thus, given ${Y_i(t) \not\in \textcolor{black}{\kappa' (\alpha_i ,t)}}$, if at some time $t$ $Y_i(t)<\frac{\alpha_i-\frac{1}{2(t+1)}}{1-\frac{1}{2(t+1)}}$, then $Y_i(t)<\frac{\alpha_i-\frac{1}{2(t+1)}}{1-\frac{1}{(t+1)}}<\alpha_i$ holds, which results in $ h_i(t)<h_i(t-1)$ on applying~\eqref{h_update_2}. The decrease of $h_i$ expands the feasible state set for $x(t+1)$ in~\eqref{mpc_general} and thereby encourages constraint violations at $t+1$ in closed-loop. While under the ideal control policy (Assumption~\ref{assumption_ideal_control_policy}), $h_i(t)<h_i(t-1)$ would have guaranteed constraint violation at time $t+1$ a.s., in practical scenarios, just the expansion of the feasible state set alone cannot guarantee violation. A similar argument can be made when $Y_i(t)>\frac{\alpha_i}{1-\frac{1}{2(t+1)}}$, which by virtue of $Y_i(t)>\alpha_i>\frac{\alpha_i-\frac{1}{2(t+1)}}{1-\frac{1}{(t+1)}}$ and~\eqref{h_update_2} leads to $h_i(t)>h_i(t-1)$, contracting the feasible state set for $x(t+1)$ thereby discouraging constraint violations at $t+1$ in closed-loop. Thus under deviation of the practical control policy from the ideal, while asymptotic convergence of $Y_i$ to $\alpha_i$ cannot be guaranteed, it still is encouraged by~\eqref{h_update_2}. }
\end{subsubsection}
\vspace{-1em}
\subsection{Asymptotic Behavior of the Practical System}\label{asmp}

It is important to determine the asymptotic behavior (i.e., as $t\to \infty$) of the practical system in which the simplifying assumption of applying an ideal control policy referred to in Assumption~\ref{assumption_ideal_control_policy} is dropped.

\begin{theorem}\label{asmp_conv}
\textcolor{black}{Let Assumptions~\ref{assumption_system},~\ref{assumption_uncertainty}, 
and~\ref{assumption_control_inputs}(a) hold.} The expected value of $Y_i(t+1)|\mathcal{F}_t$ asymptotically converges to $Y_i(t)$.
\end{theorem}

\begin{proof}\label{proof_2}
{Applying the limit of $t\to \infty$ to~\eqref{Delta_3}, and rearranging to bypass the indeterminate $\frac{\infty}{\infty}$ form, we get,
\vspace{-5pt}
\begin{equation}\label{asymptotic violations} 
\begin{aligned}
\lim_{t\to\infty} \Delta_i(t)  = &\lim_{t\to\infty} \ p_i(t+1)\Biggl[\left|\alpha_i - \frac{Y_i(t)}{1+\frac{1}{t}} -\frac{1}{t+1}\right|-\\
&\left|\alpha_i - \frac{Y_i(t)}{1+\frac{1}{t}}\right|\Biggr] + \left|\alpha_i - \frac{Y_i(t)}{1+\frac{1}{t}}\right| -
|\alpha_i - Y_i(t)|.
\end{aligned}
\end{equation} 
As $p_i(t\!+\!1) \in [0,1]$, substituting $t\to\!\infty$ in~\eqref{asymptotic violations}, yields $\lim_{t\to\infty} \Delta_i(t)\!=\!0$, which implies either: 

\begin{case}{1}\label{case_i}
\emph{$\lim_{t\to\infty} Y_i(t)=\lim_{t\to\infty}\mathbb{E}[Y_i(t+1)|\mathcal{F}_t]$.}
\end{case}

\begin{case}{2}\label{case_ii}
\emph{$\lim_{t\to\infty} Y_i(t)=\alpha_i+l$ and $\lim_{t\to\infty} \mathbb{E}[Y_i(t+1)|\mathcal{F}_t]=\alpha_i-l$, where, $l\in[-\alpha_i,\alpha_i]$.} 

\noindent \emph{Taking expectation on both sides of~\eqref{Y(t+1)} yields,}
\begin{equation} \label{Y(t+1)_proof}
\begin{aligned}
\mathbb{E}[Y_i (t+1)|\mathcal{F}_t] = t\frac{Y_i(t)}{t+1} +\mathbb{E}\biggl[\frac{V_i(t+1)|\mathcal{F}_t}{t+1}\biggr].
\end{aligned}
\end{equation}

\noindent \emph{Taking the limit at ${t\to\infty}$ in~\eqref{Y(t+1)_proof}, and substituting the values of $\lim_{t\to\infty} Y_i(t)$ and $\lim_{t\to\infty} \mathbb{E}[Y_i(t+1)|\mathcal{F}_t]$ yields,}
\begin{equation} \label{_proof_final}
\begin{aligned}
\alpha_i-l =\lim_{t\to\infty} \frac{\alpha_i+l}{1+\frac{1}{t}}+\lim_{t\to\infty}\mathbb{E}\biggl[\frac{V_i(t+1)|\mathcal{F}_t}{t+1}\biggr].
\end{aligned}
\end{equation}

\noindent \emph{As $\mathbb{E}[V_i(t+1)|\mathcal{F}_t]=p_i(t+1) \in [0,1]$, hence evaluating the limit in~\eqref{_proof_final} yields $l=0$. Hence Case~\ref{case_ii} implies Case~\ref{case_i} (but not the other way around).}
\end{case}

Case~\ref{case_i} completes the proof.} 
\end{proof}

\textcolor{black}{
The above proof shows that the time-average of the state constraint violations has a martingale-like behavior asymptotically which may be useful for practical operation of the proposed OA-SMPC to avoid unpredictable violation behavior in the long run as the system evolves.} 


\subsection{Post-processing for Real-time System Operation}\label{post-process}

In previous works such as~\cite{IEEE_1}, after the nominal MPC computed optimal control inputs for the first time-step of the prediction horizon are implemented, the observed states get corrected in closed-loop to account for real-time uncertainties by~\eqref{dynamics_real}. Accommodation of the entire uncertainty (uncertainty in weather forecast) by the state is reasonable for building climate control applications where the state (room temperature) and control input (heating/cooling effect from the air-conditioner) are not coupled to the same physical equipment~\cite{IEEE_1}.
However, in certain applications where the states and control inputs are coupled to the same equipment (like BESS, where state of charge (SOC) is the state, and charging/discharging power is a control input), the first time-step optimal control inputs (computed by the nominal MPC) can also be post-processed in real-time to correct for the realized uncertainty. For example, the BESS can alter its nominal MPC computed optimal control inputs and states to correct for the VRES and gross load forecast uncertainty in real-time to maintain power balance of the MG with the main grid.


\begin{assumption}[\textbf{Post-processing to correct for uncertainties in real-time}]\label{assumption_post_process}
{For each chance constrained state affected by uncertainties: (a) There are two mutually coupled sources of control with one source having the primary responsibility of handling the uncertainties in closed-loop. During correction of the observed states in closed-loop to account for the uncertainty, the feasible altered control input and state set is the same as defined by {\normalfont ${\textbf{\text s}_{1:q}}$}~\eqref{input_compact_nominal} and {\normalfont $\textbf{\text g}_{1:r}-\textbf{\text h}_{1:r}(t)$}~\eqref{chance_constraints_compact_nominal} respectively, which are the time-varying design limitations. 
The secondary control source always has enough control input available to handle the remaining part of the uncertainty (through satisfaction of~\eqref{equal_real}) that cannot be handled by the primary source due to the possibility of violation of the time-varying design limitations by the primary control source in closed-loop, (b) The state transition is not dependent on the secondary control source, (c) The primary and secondary control sources are coupled only to each other, (d) The state transition is not dependent on the control sources associated with other states.}
\end{assumption}
${\textbf{\text s}_{1:q}}$ is a hard constraint and does not vary with time, but as $\textbf{\text h}_{1:r}(t)$ is time-varying, we refer to these design limitations as time-varying, when considered together. After updating the states and control inputs, $h(t+1)$ is updated by~\eqref{h_update_2}, and the optimization in~\eqref{mpc_general} is repeated. The schematic diagram of the complete OA-SMPC operational framework with post-processing is shown in Fig.~\ref{MPC_framework}, with the algorithm presented in {Algorithm}~\ref{alg:alg1}. 
The results from Theorems~\ref{monotonic_conv},~\ref{just_conv} and~\ref{asmp_conv} are independent of the post-processing framework, and still hold with post-processing. Additionally, the restrictive one step controllability in Assumption~\ref{assumption_control_inputs}(b) can be relaxed for practical systems while still maintaining the structure of the input matrix in Remark~\ref{Recursive} due to Assumption~\ref{assumption_post_process}. 

\textcolor{black}{Note that the computational cost of the proposed OA-SMPC is same as that of a nominal MPC, with the $h$ update happening outside of the MPC framework which makes the present method extremely scalable for practical implementation. Also, similar to~\cite[Sec. VII]{IEEE_1}, as the linear structure of the system is not leveraged for Theorems~\ref{monotonic_conv},~\ref{just_conv} and~\ref{asmp_conv}, the results hold for nonlinear systems too provided 
the relevant assumptions hold. }
%

\vspace{-2em}
\begin{figure}[ht]
\centering
\includegraphics[scale=0.167]{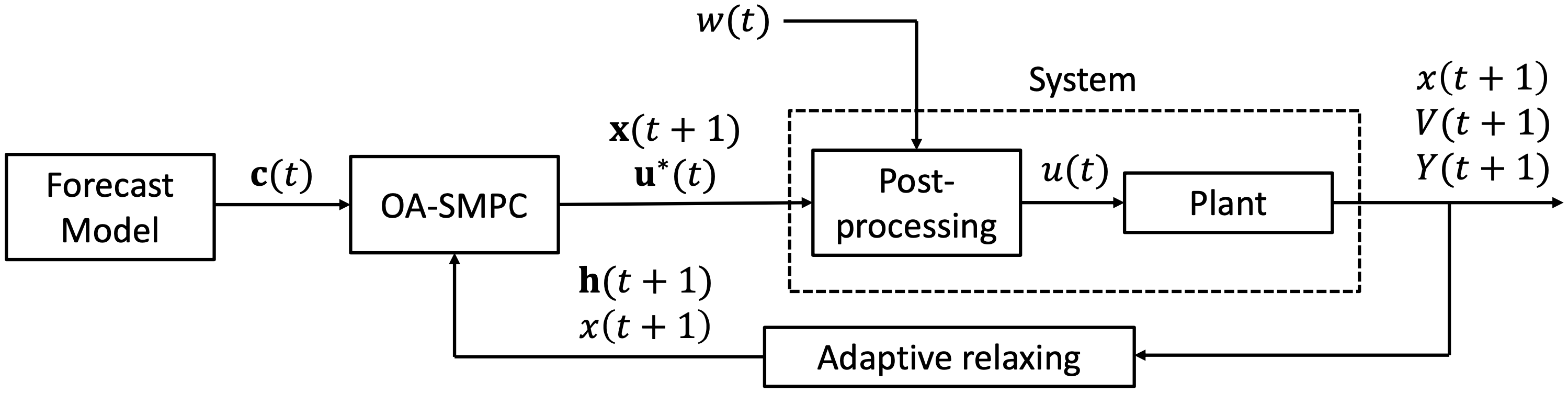}
\caption{OA-SMPC operational framework with post-processing.}
\vspace{-1em}
\label{MPC_framework}
\end{figure}


\begin{example}
    
\emph{To demonstrate post-processing, consider a simple system with $x(t)=[x_1(t)]$ as the state subjected to chance constraints, control inputs $\setlength\arraycolsep{3pt} u(t)\!=\!{\begin{bmatrix} u_1(t) &  u_2(t)\end{bmatrix}}^{\top}$, where $u_1$ and $u_2$ are the primary and secondary control sources respectively, and the uncertainty is denoted by $w(t)$. The system matrices are $A$, $\setlength\arraycolsep{3pt} B\!=\!{\begin{bmatrix} B_{11} &  0\end{bmatrix}}$ and $E$. The structure of $B$ follows from {Assumption}~\ref{assumption_post_process}(b). $u_1(t)$ and $u_2(t)$ are coupled by~\eqref{equal_real}, following {Assumption}~\ref{assumption_post_process}(c). Let $\tilde{x}_1(t+1)$ and $\tilde{u}_1$  be the state and primary control upper limits based on the time-varying design limitations following {Assumption}~\ref{assumption_post_process}(a). The system described is similar to a MG with BESS, VRE, local load demand and grid connectivity, where the BESS SOC is the state, BESS charging/discharging power is the primary control source, and the grid import power is the secondary control source. Both the control sources are coupled by the grid power balance equation, while the uncertainty is the real-time VRES and load forecast error.}

\emph{When correcting for the uncertainties in closed-loop, first it is ensured that the primary control source handles only the part of the uncertainty that still keeps it within the feasible set. The altered primary control input can be formulated as,} 
\vspace{-2pt}
\begin{subequations}\label{u_altered}
\begin{equation} \label{u_altered_1}
\begin{aligned}
{u}_1(t):= \min \Bigl(u_1^*(t|t)+D_1w(t),\tilde{u}_1\Bigr),
\end{aligned}
\end{equation}
\emph{where $D = B^{\dagger}E$, and $u_1^*(t|t)$ is the optimal MPC computed primary control input. The state update equation using {Assumptions~\ref{assumption_post_process}(a) and~\ref{assumption_post_process}(b)} is formulated as} 
\begin{equation} \label{x_altered}
\begin{aligned}
x_1(t+1) =\min \Bigl(Ax_1(t)+B_{11}{u}_1(t), \tilde{x}_1(t+1) \Bigr).
\end{aligned}
\end{equation}
\end{subequations}
\emph{If from~\eqref{x_altered}, $x_1(t+1)=\tilde{x}_1(t+1)$, we re-compute, ${u}_1(t)= (B_{11})^{-1} \left(\tilde{x}_1(t+1)-Ax_1(t)\right)$. Finally, the  altered secondary control input $u_2(t)$ is computed by satisfying~\eqref{equal_real}. In the case of the time-varying design limitations giving the lower limits of the state and primary controls, $\vec {x}_1(t+1)$ and $\vec {u}_1$ respectively,~\eqref{u_altered_1} and~\eqref{x_altered} are modified by replacing the $\min$ by the $\max$ function. Assumptions~\ref{assumption_post_process}(a) and~\ref{assumption_post_process}(b) ensure that the post-processing steps give unique solutions, and Assumptions~\ref{assumption_post_process}(c) and {\ref{assumption_post_process}(d)} ensure that the correction for uncertainties affecting a chance constrained state and related control inputs does not unnecessarily affect other states and control inputs which may lead to inconsistency 
in the post-processed solutions. 
After post-processing, the state constraint violations in closed-loop are tracked as~\eqref{V_post_processed} with time-average of violations calculated similar to~\eqref{Y_general}.}
\begin{equation} \label{V_post_processed}
\begin{aligned}
\!\!\!\!V_1(t+1)\!:= \!\left\{ 
\begin{matrix}
1,  &\!\!\!G_1x(t+1)>g_1,\\
0,  &\!\!\!G_1x(t+1)\leq g_1.
\end{matrix}
\right.
\end{aligned}
\end{equation} 
\end{example}

\begin{remark}[\textbf{Significance of the post-processing framework}]\label{uncertainty}
{
In related previous works with two mutually coupled sources of control to handle uncertainties on chance constrained states such as~\cite{Guo2018}, the time-varying design limitations as that of \textit{Assumption~\ref{assumption_post_process}}(a) are not considered. Large primary control inputs to handle large uncertainties can, thus, potentially lead to damaging the primary controller in~\cite{Guo2018}. In other works like~\cite{IEEE_1}, MPC computed control inputs are not altered to account for uncertainties, making the state handle the entire uncertainty in closed-loop. Large uncertainties can thus steer the system away from the OA-SMPC feasible region, which due to \textit{Assumption~\ref{assumption_control_inputs}}(a) can lead to the application of 
expensive control input at the next time-step to bring the system back to feasibility. Such an expensive control input can lead to high economic cost, a problem not considered in both~\cite{IEEE_1} and~\cite{Guo2018}.} 
    
\end{remark}

\begin{remark}[\textbf{\textcolor{black}{Non-conservative chance constraint satisfaction in a `practical sense' in closed-loop even under time-varying uncertainty distribution, and repeated large uncertainties}}]\label{practical}
{
As the nominal OA-SMPC always has access to relaxed feasible states, it is possible for the predicted nominal solutions to always violate the (original) state constraints $\forall k,t$ (as satisfaction of~\eqref{chance_mpc_previous_1} is not mandatory in our formulation). In closed-loop, thus, if solutions are continuously violated more than the maximum prescribed level, $h_i$ increases until $h_i \rightarrow 0^-$ according to~\eqref{h_update_2}. Mathematical violations can still persist in $V_i(t+1)$ after $h_i \rightarrow 0^-$ which fails to give the benefit of adaptive relaxation. These violations can be ignored during practical implementation, by adding a small $\varepsilon >0$ to $g_i$ resulting in considering violations only if $G_ix(t+1)>g_i+\varepsilon$. The consideration of $\varepsilon$ in tracking violations is an operational step and can be removed by the user when $h_i$ decreases sufficiently to relax the state constraint and reduce conservatism.
\textcolor{black}{This adaptive behavior of our system along with post-processing ensures that the chance constraints in~\eqref{chance_real} are satisfied in a `practical sense' in closed-loop, which previous works~\cite{IEEE_1, Guo2018,capone2024online} may not be able to satisfy
under repeated large uncertainties
when the mixed worst-case disturbance sequence is not known a-priori, or if the uncertainty distribution changes with time.} 
Note that the mixed worst-case disturbance sequence computed via scenarios of uncertainty in~\cite{IEEE_1,Guo2018} may in itself be over-conservative~\cite{lorenzen2016constraint}.}

\end{remark}

\begin{algorithm}[]
\caption{Online Adaptive SMPC (OA-SMPC)}\label{alg:alg1}
\begin{algorithmic}
\STATE 
\STATE {{Initialization}}
\begin{enumerate} \item[1.] Choose $x(0)$.
\item[2.] Choose $h(0)$ from domain knowledge.
\end {enumerate} 
\STATE {{Online solution}}
\begin{enumerate} \item[1.] Solve~\eqref{mpc_general} to get $u^*(t|t)$.
\item[2a.] If post-processing of nominal OA-SMPC computed optimal control inputs are not allowed: Set \hbox{$u(t)\!\gets\!u^*(t|t)$} and account for the uncertainties in $x(t\!+\!1)$ by~\eqref{dynamics_real}. Calculate $V(t+1)$ from~\eqref{V_general}.
\item[2b.] If post-processing of nominal OA-SMPC computed optimal control inputs are allowed: Post-process by~\eqref{u_altered} and~\eqref{equal_real} to account for the uncertainties to calculate $x(t+1)$ and $u(t)$. Calculate $V(t+1)$ similar to~\eqref{V_post_processed}. 
\item[3.] Calculate $Y(t+1)$ from~\eqref{Y_general}.
\item[4.] Calculate $h(t+1)$ from~\eqref{h_update_2}.
\item[5.] Set $t \gets t+1$ and repeat from Step 1 of Online solution.
\end {enumerate}
\end{algorithmic}
\label{alg1}
\end{algorithm}

\vspace{-1em}
\section{Case study}\label{case_study}
\subsection{Overview} \label{overview}
In this section, we implement our proposed method (OA-SMPC) for simulating the optimal BESS dispatch strategy for a practical MG with PV, load and connection to the main grid \textcolor{black}{in an EMPC framework.} The MG setup is from the real-life MG at the Port of San Diego, described in~\cite{ghosh2023adaptive}. The MG model incorporates electricity prices with demand and energy charges, realistic load and PV forecast, and a post-processing step for incorporating the forecast uncertainties in BESS dispatch. The yearly (2019) electricity costs were compared for the MG, for the traditional MPC, with hard constraints on the state (BESS SOC), and our OA-SMPC method with chance constraints on the state. The motivation for using chance constraints on BESS SOC in the OA-SMPC is to leverage some extra BESS capacity to reduce demand peaks and thus, demand charges, leading to significant electricity cost savings, while staying within a maximum violation probability bound to avoid adverse effects on BESS life. \textcolor{black}{Additionally, we also compare our proposed OA-SMPC to the traditional EMPC method without chance constraints, and a state-of-the-art approach~\cite{IEEE_1} from the literature with chance constraints and similar computational cost.}
\subsection{PV and Load Forecast} \label{forecast}

The MG model uses day-ahead PV and gross load forecasts with 15 minute time resolution as inputs. The k-Nearest Neighbor (kNN) algorithm is used for the gross load forecast. The training data comprises of $15$-min resolution historical load observations from November $1$, $2018$ to November $20$, $2019$. 
For every MPC horizon, gross load observations for the previous $24$~h (feature vector) are compared with the training sample and $k\!=\!29$ nearest neighbors are identified by the kNN algorithm. Finally, the gross load forecast is calculated by averaging the gross load of the selected neighbors at every time-step of the forecast horizon~\cite{Zhang2016}. \textcolor{black}{The root mean square error (RMSE), mean absolute error (MAE), and mean bias error (MBE) for the gross load forecast for the entire year are $22.4$~kW, $17.3$~kW, and $-1.2$~kW respectively.} 

The PV generation forecast for the upcoming $24$~h utilizes the kNN (with $k\!=\!30$) algorithm as well. 
The feature vector is formed by three data sets: numerical weather prediction (NWP) model forecasts for the upcoming $24$~h, the average of the preceding $1$, $2$, $3$, and $4 $~h PV power generation, and the current time of the day. The NWP forecast utilized was the High-Resolution Rapid Refresh (HRRR) model developed by NOAA~\cite{HRRRpaper}. 
The PV power generation dataset was obtained 
by running simulations for the PV plant in the Solar Advisor Model (SAM) 
using the irradiance observation data obtained from the NSRDB database as an input. The training sample for PV power generation forecast includes NWP forecast and irradiance observations between January $1$, 2019 to December $31$, 2019. \textcolor{black}{The RMSE, MAE and MBE for the PV forecast for the entire year are $20.3$~kW, $9.1$~kW, and $-0.8$~kW respectively.}


\subsection{Microgrid (MG) Model}\label{MG_model}
The system state $x(t)=$\:[$x_{1}(t)$] is the BESS state of charge (SOC). The control input is $\setlength\arraycolsep{3pt} u(t)\!=\!{\begin{bmatrix} u_1(t) &  u_2(t)\end{bmatrix}}^{\top}$ where $u_1(t)$ is the BESS dispatch power (primary control source for handling uncertainty), and $u_2(t)$ is the grid import power (secondary control source for handling uncertainty). $u_{1}(t)>0$ denotes charging, while $u_{2}(t)>0$ denotes power import from the main grid to the MG. The PV generation and gross load is denoted by $\text {PV}(t)$ and $L(t)$ respectively and are used as forecast inputs to the MPC. The uncertainty $w(t)=$\:[$w_{1}(t)$] is the difference between the gross load and PV generation forecast uncertainties, i.e., $w_{1}(t)= \bigl({L}^{{f}}(t)- {L}^{{r}}(t)\bigr) - \bigl(\text {PV}^{f}(t)-\text {PV}^{{r}}(t)\bigr) $ where the superscript $f$ and $r$ denote forecasted and real values. The MPC prediction horizon is one-day ahead, subdivided into $N=96$ equal time-steps of $\Delta t=0.25$ h (15 minutes) each.

The system matrices are $A=[1]$, $B={\begin{bmatrix} \frac{\Delta t}{\text{BESS}_\text{en}} & \!0 \end{bmatrix}}$, and $E=[\frac{\Delta t}{\text{BESS}_\text{en}}]$ where $\text{BESS}_\text{en}$ is the energy capacity of the BESS. The system matrices handle the SOC update of the battery due to charging/discharging. For the hard control input constraints, $S=\begin{bmatrix}
1 & 0 \\
-1 & 0 
\end{bmatrix}$ and $\setlength\arraycolsep{3pt} s={\begin{bmatrix} \text {BESS}_{\text {max}} & \text {BESS}_{\text {max}} \end{bmatrix}}^{\top}$, which constrains the maximum charging/discharging power of the BESS. For the time varying equality constraints coupling the control inputs, $\setlength\arraycolsep{3pt} M={\begin{bmatrix} 1 & -1 \end{bmatrix}}$, $ \textcolor{black}{c(t)}=[\text{PV}^f(t)-L^f(t)]$, and $F=[1]$, which ensures power balance of the MG with the main grid. 
The MPC also has a terminal state constraint defined as $x_1(t\!+\!N|t)\geq\hat{x}_1$.

For this case study, we consider joint chance constraints (JCC) on the BESS SOC which considers a violation if the BESS SOC goes above or below predefined upper ($\text {SOC}_{\text {max}}$) and lower bounds ($\text {SOC}_{\text {min}}$) in closed-loop. The goal is to reduce electricity import costs from the grid by having controlled violations beyond the predefined upper and lower bounds by making a larger BESS capacity available for dispatch. Unrestricted violations are avoided as they can adversely affect the BESS lifetime. The chance constraints are defined by $G={\begin{bmatrix} 1 & \!-1 \end{bmatrix}}^{\top}$, $\setlength\arraycolsep{3pt} g={\begin{bmatrix} \text {SOC}_{\text {max}} & -\text {SOC}_{\text {min}} \end{bmatrix}}^{\top}$ and $\setlength\arraycolsep{3pt} h(t)={\begin{bmatrix} h_{1}(t) &  h_{2}(t)\end{bmatrix}}^{\top}$. We choose $h_1(t)=h_2(t)$ for adapting both the state constraints simultaneously by the same parameter as they are setup in JCC form. The maximum probability of the JCC violation is predefined by $\alpha$. For the traditional \textcolor{black}{EMPC} (without chance constraints), violations are avoided and the state constraints are formulated as~\eqref{trad_state}, while for the OA-SMPC the state constraints are formulated as~\eqref{chance_mpc_now}.
\begin{equation} \label{trad_state}
\begin{aligned}
Gx(t\!+\!k|t) \leq g \qquad \forall k\in \mathbb{N}_{1}^{N},\,\forall t.
\end{aligned}
\end{equation}

In the JCC formulation, when updating $h(t)$ by~\eqref{h_update_2}, it is ensured that $h_1(t)=\max\bigl(\text{SOC}_{\text {max}}\!-\!1,h_1(t)\bigr)$, and $h_2(t)=\max\bigl(-\text{SOC}_{\text {min}},h_2(t)\bigr)$ to ensure the state constraints do not violate physical limits of SOC above 1 or below 0. \textcolor{black}{However, in our case study, $h_i(t)$ 
given by~\eqref{h_update_2} never violate physical limits, obviating the above correction. We have practically ensured this by setting a high value of $\gamma_i$ in~\eqref{h_update_2}, which is a design choice, at the cost of slower system adaptation (i.e., rate of change of $h_i(t)$), and setting the initial constraint relaxation parameter $h_i(0)$ such that $g_i-h_i(0)$ is sufficiently far away from physical limits for $i \in \{1,2\}$.}

The objective function is formulated as in~\cite{ghosh2023adaptive,ghosh2022effects}, and is given by,
\begin{align} \label{objective}
&J(t)\!=\! R_\text{NC}\! \max\{u_2(t\!+\!k|t)\}_{k=0}^{N\!-\!1}\!+\!
R_\text{OP}\! \max\{u_2(t\!+\!l|t)\}_{l \in \mathbb{I}(t)}  \nonumber \\
&+ R_{\text{EC}}\Delta t\Biggl[\sum_{k=0}^{N-1}{u_2}(t\!+\!k|t) + \frac{1-\eta}{2}\sum_{k=0}^{N-1}|{u_1}(t\!+\!k|t)|\Biggr],
\end{align}
where ${R_{\text{NC}}}$ is the non-coincident demand charge (NCDC) rate charged on the maximum grid import during the prediction horizon. Similarly, ${R_{\text{OP}}}$ is the on-peak demand charge (OPDC) rate, charged on the maximum grid import between $16$:$00$ and $21$:$00$ h, called on-peak (OP) hours of the prediction horizon, and $\mathbb{I}(t)$ represents indices of prediction horizon time-steps coinciding with the OP hours. Naturally, the indices of OP hours in the prediction horizon are a function of the starting time-step $t$ of the MPC prediction horizon. ${R_{\text{EC}}}$ is the energy charge rate and $\eta$ is the round-trip efficiency of the BESS accounting for BESS losses. After the net load for the entire month is realized, the monthly NCDC is computed based on maximum load demand from the grid during the month, while the monthly OPDC is computed based on maximum load demand between $16$:$00$ and $21$:$00$~h of all days of the month. The predefined parameters of the MG for the OA-SMPC operation are shown in Table~\ref{parameters} and a block diagram of the MG operational framework is shown in Fig.~\ref{MPC_framework}. \textcolor{black}{Note that the one step controllability in Assumption~\ref{assumption_control_inputs}(b), and ideal control policy in Assumption~\ref{assumption_ideal_control_policy} is relaxed for the case study for realistic simulations.} The satisfaction of Assumption~\ref{assumption_control_inputs}(a) is ensured by choosing a sufficiently large $\gamma_i$ which constrains the rate of $h_i(t)$ increase depending on the available control input power to satisfy~\eqref{eq_3}. Note that due to Assumption \ref{assumption_post_process}(a), in our case study, the risk of violation of Assumption~\ref{assumption_control_inputs}(a) can only arise when $h_i(t)$ increases.\footnote{\textcolor{black}{The maximum absolute forecast uncertainty throughout the year for the case study is 261 kW, which still is lesser than half of the BESS power capacity of 700 kW. Thus, even in the case where Assumption~\ref{assumption_post_process}(a) is relaxed, and the BESS is set up to handle all the uncertainty until reaching the physical limits, the practical viability of Assumption~\ref{assumption_control_inputs}(a) is reinforced.}}
 \vspace{-1em}

\begin{table}[ht]
\centering
\caption{\textcolor{black}{Design Parameters of the OA-SMPC for the MG.}}
\resizebox{\columnwidth}{!}{%
{\begin{tabular}{lcc}
\cline{1-3}
Parameter            & Symbol              &  Value  \\ 
\cline{1-3}
NCDC  rate           & $R_\text{NC}$       & \$24.48/kW        \\
OPDC  rate           & $R_\text{OP}$       & \$19.19/kW       \\
Energy rate          & $R_\text{EC}$       & \$0.1/kWh \\
BESS round-trip efficiency      & $\eta$                &  0.8  \\
BESS energy capacity       & $\text {BESS}_\text{en}$     & 2,500 kWh \\
BESS power capacity  & $\text {BESS}_\text{max}$        &700 kW     \\
Upper bound of SOC for \\
traditional \textcolor{black}{EMPC 1}  & $\text {SOC}_\text{max}$          & 0.8    \\
Lower bound of SOC for \\
traditional \textcolor{black}{EMPC 1}   & $\text {SOC}_\text{min}$          & 0.2    \\
Maximum violation probability       & $\alpha$     & 0.1 \\
Terminal state constraint  & $\hat{x}_1$        &       0.5     \\
Initial state  & $x_1(0)$        &       0.5     \\
Initial constraint relaxing parameter  & $h(0)$        &       $\setlength\arraycolsep{3pt}\!{\begin{bmatrix} -0.1 & \! -0.1\end{bmatrix}}^{\top}$      \\
Proportionality constant       & $\gamma_1$ and $\gamma_2$     & 15 \\
\cline{1-3}\\
\end{tabular}}
}
\label{parameters}
\vspace{-2em}
\end{table}
\vspace{-0.25em}
\subsection{Operation Strategy}\label{strategy}

\textcolor{black}{This section presents the real-time operation strategy of the OA-SMPC for the economic MG dispatch probem under consideration.} 
Although~\eqref{h_update_2} updates $h_i$ at every time-step, given the emphasis on additional OPDC penalties in our cost function, whenever the starting time-step of the MPC coincides with daily OP hours, we restrict the $h_i$ increase between two time steps, overriding~\eqref{h_update_2} when required. Decreasing $h_i$ during the daily OP hours is still allowed, and if $h_i$ increases it is reset manually to the previous value. The reasoning behind restricting the increase of $h_i$ during OP times is to avoid the additional OPDC costs (on top of NCDC) by incurring preferential violations during OP hours (by deeper BESS discharge due to relaxed state constraints). The preferential OP hour violations may lead to state violations beyond the maximum allowable limit temporarily, but the adaptive rule (by virtue of $Y>\alpha$) pushes the BESS to compensate by lowering violations due to rapid increase of $h_i$ during other times (when risk of penalty on peaks is lower). The simulations are carried out in CVX, a package for solving convex programs in the MATLAB environment~\cite{cvx_1,cvx_2}. 

\section{Results and Discussion}\label{results}

We compare results of yearly simulations from 4 test cases: (i) Traditional \textcolor{black}{EMPC} 1 described in Section~\ref{MG_model} to demonstrate the case without chance constraints; (ii) OA-SMPC which is our proposed method described in Section~\ref{MG_model} to demonstrate the case with chance constraints on state; \textcolor{black}{(iii) SMPC Lit from ~\cite{IEEE_1} which is similar in computational cost to our proposed method, with the chance constraints being represented by~\eqref{chance_mpc_previous_2}, a corresponding adaptive constraint tightening rule given by $\tilde{h}_i(t)=\tilde{h}_i(t-1)\biggl[1-\frac{\alpha_i-Y_i(t)+\frac{2\alpha_i-1}{2t}}{\gamma_i}\biggr]$, with $\tilde{h}_i(t)>0,  \; \forall t$, $i \in \{1,2\}$, and $\tilde{h}_1(t)=\tilde{h}_2(t)$. As~\cite{IEEE_1} only tightens the state constraints, violation is allowed in closed-loop by changing the feasible state set in the post-processing step in Assumption \ref{assumption_post_process}(a) to the physical limits of the system (i.e., SOC limits of 0 and 1). Similar to the OA-SMPC, $\tilde{h}_i$ increase during OP hours is avoided by resetting it manually to its previous value. The design parameters for SMPC Lit are the same as that of OA-SMPC (see Table~\ref{parameters}), with $\tilde{h}(0)=-h(0)$}; (iv) Traditional \textcolor{black}{EMPC} 2 which modifies Case (i) with $\setlength\arraycolsep{3pt} g\!=\!{\begin{bmatrix} \text {SOC}_{\text {max}}\!-\!h_1(0) & -\text {SOC}_{\text {min}}\!-\!h_2(0) \end{bmatrix}}^{\top}$ to demonstrate the MG performance with the state constraints always relaxed by the same initial relaxing parameter in Case (ii) to violate $\text {SOC}_{\text {max}}/\text {SOC}_{\text {min}}$ in closed-loop but without adaptation.

\textcolor{black}{From the PV and gross load forecasting errors, we find that the mean and standard deviation (SD) 
of the uncertainty is $-0.5$~kW and $30.3$~kW, 
respectively for the entire year of 2019. A greater (lesser) value of mean uncertainty would push the states higher (lower) in general causing more constraint violations due to the BESS exceeding (falling below) the $\text{SOC}_{\text {max}}$ ($\text{SOC}_{\text {min}}$). A higher SD of the uncertainty would cause the closed-loop behavior of the system to differ significantly from the open-loop solutions given by~\eqref{mpc_general} which in addition to increasing the likelihood of violations may also compel expensive BESS dispatch to handle the uncertainty and satisfy Assumption~\ref{assumption_control_inputs}(a).} 

\vspace{-1em}
\begin{table}[ht]
\centering
\caption{Results for the 4 test cases for the year 2019. Monthly costs are added to compute the yearly cost.}
\resizebox{\columnwidth}{!}{\begin{tabular}{lrrrr}
\cline{1-5}
Costs            & Traditional 1   &  OA-SMPC &  SMPC Lit~\cite{IEEE_1} &Traditional 2   \\ 
\cline{1-5}
NCDC             & \$66,586     & \$64,579 &  \$66,960    &   \$57,690   \\
OPDC             & \$1,959      & \$1,274  &  \$2,087    &    \$1,148    \\
Energy Cost      & \$14,382     & \$14,385 &  \$14,382     & \$14,399\\
BESS loss        & \$8,290      & \$9,071  &  \$8,309   &  \$10,447  \\
Total Cost       & \$91,217     & \$89,309 &  \$91,738     &  \$83,683 \\
Total BESS cycles &165.8        &181.4         &166.2   &208.9   \\
$Y$ at year end    & 0\%          & 10.1\%    & 2.2\%  & 20.4\%   \\
\cline{1-5}\\
\end{tabular}}
\label{Year_Table}
\vspace{-3em}
\end{table}

\begin{figure}[ht]
\includegraphics[scale=0.082]{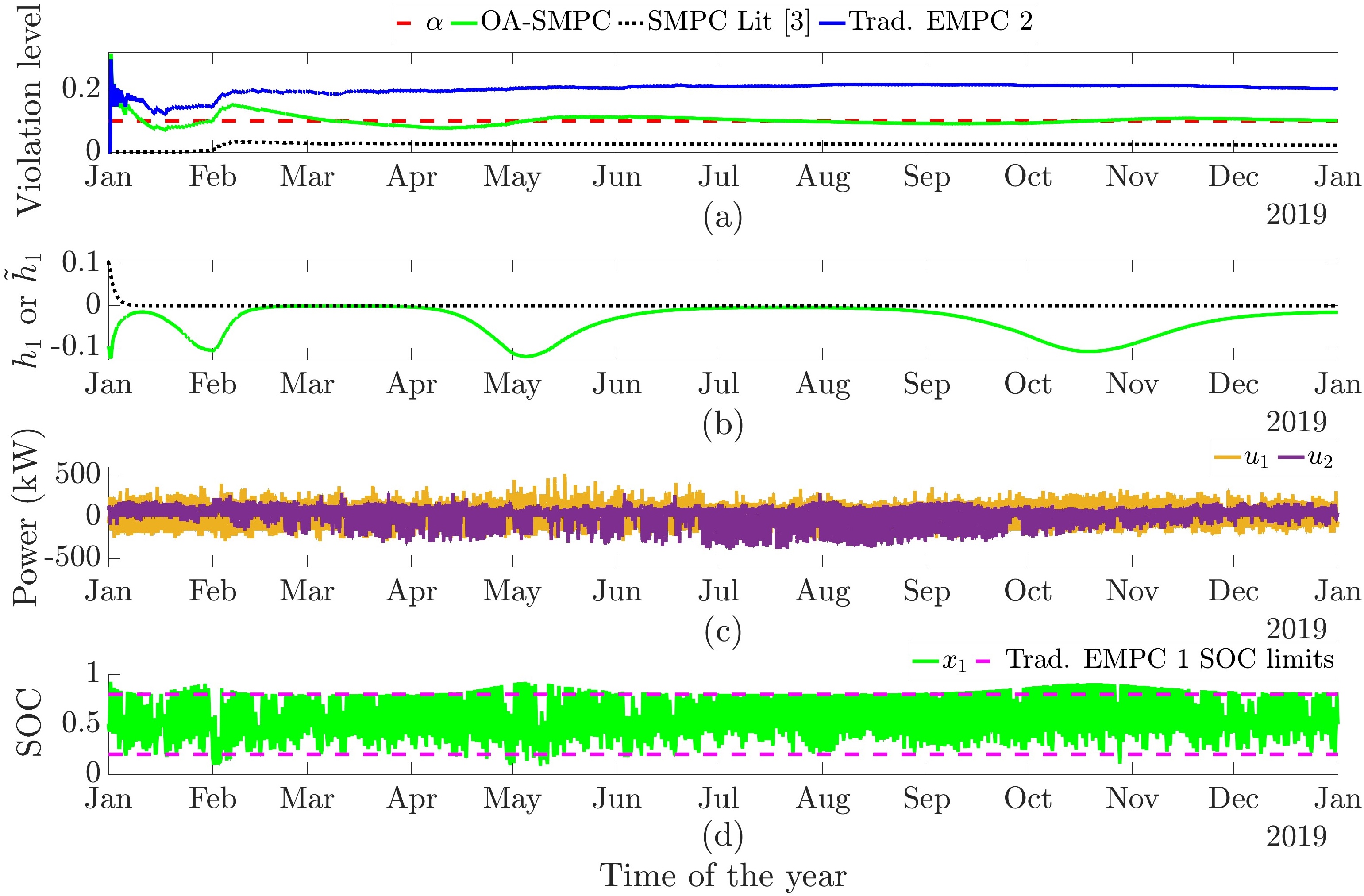}
\vspace{-2em}
\caption{\textcolor{black}{{Yearly time-series for the: (a) time-average of state constraint violations ($Y$) for the OA-SMPC, SMPC Lit~\cite{IEEE_1}, and Traditional EMPC 2 case studies, and the maximum allowable state constraint violation probability ($\alpha=0.1$), (b) adaptive state constraint relaxing parameters ($h_{1}$ and $\tilde{h}_1$, with $h_1\!=\!h_2$, and $\tilde{h}_1=\tilde{h}_2$ in these case studies) for the OA-SMPC and SMPC Lit, (c) closed-loop behavior of the control inputs for the OA-SMPC, (d) closed-loop behavior of the state for the OA-SMPC.}}}
\label{Year_Figure}
\end{figure}

Figure~\ref{Year_Figure} and Table~\ref{Year_Table} summarize the simulation results of the 4 test cases for the year 2019. Note that the energy costs are similar across all cases as there is no arbitrage in the MG model. The minute differences in energy costs are due to the different ending BESS SOC at the end of the year for the different test cases. 

The comparison between the \textcolor{black}{Traditional EMPC 1 and OA-SMPC} demonstrates the superior economic performance of our proposed algorithm 
while still staying within the maximum allowable violation probability bound. \textcolor{black}{NCDC and OPDC decrease by 3\% and 35\% between Traditional EMPC 1 and OA-SMPC.} The OPDC savings (as a \%) are significantly more than NCDC because of the operation strategy (see Section~\ref{strategy}) allowing preferential violations during OP hours. \textcolor{black}{As extra BESS capacity is available for the OA-SMPC as compared to the Traditional EMPC 1 due to the adaptive relaxation of BESS SOC constraints beyond the $0.2-0.8$ SOC range, the BESS is used more aggressively in OA-SMPC resulting in 9.4\% more BESS losses amounting to 15.6 extra BESS yearly cycles. Overall, OA-SMPC leads to 2.1\%  total yearly cost savings as compared to the Traditional EMPC 1. SMPC Lit performs the worst as it tightens the feasible SOC range and leads to highest yearly cost. Some violations are caused due to the effect of uncertainty in the closed-loop but it is unable to reduce costs as the timing of the uncertainty is critical in reducing demand charges. The NCDC and OPDC are greater in SMPC Lit as compared to both the Traditional EMPC 1 and OA-SMPC. SMPC Lit has similar BESS cycles as the Traditional EMPC 1 with 0.6\% extra yearly costs. This serves as a proof of the concept alluded to in Remark~\ref{Approx_Chance_Real} that adaptive constraint tightening methods (like SMPC Lit) can be over-conservative being unable to exploit the allowable violation limit which can be particularly disadvantageous in EMPC frameworks.} 

\textcolor{black}{Figure~\ref{Year_Figure}(a) shows the variation of $Y$ for the OA-SMPC (green line), and SMPC Lit (black dotted line) for the entire year. For the OA-SMPC, initially, $Y$ is $0$ until the first constraint violation occurs, after which it overshoots $\alpha$, oscillating with a high frequency and magnitude until the first week of January. Then, consistent with the goal of $Y$ converging to a value less than or equal to $\alpha$, $Y$ decreases and oscillates about $\alpha$ with lower frequency and magnitude, as the system evolves, and finally reaching a value of 10.1\% at the end of the year (which is within a small margin of $\alpha$). It is expected for $Y$ to decrease below $\alpha$ if the OA-SMPC is run over a longer time period, resulting in satisfying the chance constraints, albeit non-conservatively, in closed-loop.} 
Figure~\ref{Year_Figure}(b) shows that for the \textcolor{black}{OA-SMPC}, $h_1$ and $h_2$ increase and decrease depending on the overshoot and undershoot of $Y$ with respect to $\alpha$ respectively, being able to expand and contract the feasible state set accordingly to keep $Y$ near $\alpha$. \textcolor{black}{Figures~\ref{Year_Figure}(c) and~\ref{Year_Figure}(d) demonstrate the yearly closed-loop behavior of the BESS dispatch and grid import, and SOC respectively for the OA-SMPC. For a significant portion of the year, the grid import is negative because the PV system at the location is oversized generating excess power which is fed back to the grid. Significant violation of constraints can be seen in Figure~\ref{Year_Figure}(d) in the months of May and October, which are the two months that aid most is cost savings.} 

\textcolor{black}{Figure~\ref{Year_Figure}(b) also shows that for the SMPC Lit, the $\tilde{h}_1$ and $\tilde{h}_2$ quickly drop down to 0 in the first week of January trying to encourage violations initially as the initial constraint tightening was too conservative. However, the system's $Y$ still stays far below $\alpha$, with the result that $\tilde{h}_1$ and $\tilde{h}_2$ stay close to $0^+$ throughout the year behaving essentially like the Traditional EMPC 1 in the nominal MPC with allowance for SOC limits to range from 0 to 1 in closed-loop. In the months of February and March (not shown in Figures), SMPC Lit creates more NCDC than Traditional EMPC 1, because just before the NCDP time, a high uncertainty forces the SOC to go below the lower limit of 0.2, which causes a large BESS charging action at the next time-step to climb back to the feasible SOC range of above 0.2 causing demand peaks, a problem alluded to in Remark~\ref{uncertainty}.    }

Traditional \textcolor{black}{EMPC 2} demonstrates that relaxing the state ($\text{BESS SOC}$) constraints by the same initial relaxation parameter as in \textcolor{black}{OA-SMPC} without any adaptation causes the violations to exceed the maximum allowable violation probability bound in closed-loop (blue line in Fig.~\ref{Year_Figure}(a)). The Traditional \textcolor{black}{EMPC 2} has lower total yearly electricity costs (due to lower NCDC and OPDC), than the \textcolor{black}{OA-SMPC} because the objective function in~\eqref{objective} penalizes peaks in grid import power ($u_2$) more than BESS dispatch power ($u_1$), resulting in prolonged aggressive dispatch of the BESS to serve the net load ($L-\text{PV}$). The aggressive BESS dispatch is more prolonged in Traditional \textcolor{black}{EMPC 2} than \textcolor{black}{OA-SMPC} because the feasible BESS SOC range in Traditional \textcolor{black}{EMPC 2} is not adapted due to past violations, and the feasible state set continues being larger than \textcolor{black}{OA-SMPC} throughout. The Traditional \textcolor{black}{EMPC 2} shows higher frequency of oscillation and magnitude of overshoot of $Y$ above $\alpha$ as compared to \textcolor{black}{OA-SMPC}, and additionally shows that $Y$ consistently remains above $2\alpha$ for the majority of the year. Thus, the Traditional \textcolor{black}{EMPC 2} is unable to fulfill the chance constraints and may significantly harm BESS life, resulting in $27.5$ more yearly BESS cycles as compared to \textcolor{black}{OA-SMPC}. The analysis serves as a proof of concept that it is the adaptive relaxation rule in \textcolor{black}{OA-SMPC} rather than the nature of uncertainties which ensure non-conservative chance constraint satisfaction in closed-loop in \textcolor{black}{OA-SMPC}.

\textcolor{black}{Figure~\ref{Year_Figure_alpha} and Table~\ref{Year_Table_alpha} present the results of the OA-SMPC analysis when it is repeated with different values of $\alpha\in\{0.05, 0.15, 0.2\}$, with the same other design parameters as in Table~\ref{parameters}. Figure~\ref{Year_Figure_alpha} shows that in all the cases, similar behavior is observed as in Fig.~\ref{Year_Figure}(a) with $Y$ oscillating about $\alpha$ with lower frequency and magnitude as the system evolves with time along the year while ultimately settling at a value below $\alpha$ at year-end, thereby tightly satisfying the chance constraints. Table~\ref{Year_Table_alpha} demonstrates that superior cost savings are attained when $\alpha$ increases at the cost of more BESS cycles. Comparing the case with $\alpha=0.2$ in Table~\ref{Year_Table_alpha} with Traditional EMPC 2 in Table~\ref{Year_Table} shows that the OA-SMPC is more effective at reducing costs with less violations and BESS cycles than the Traditional EMPC 2 due to the OA-SMPC's online adaptivity.}

\vspace{-1em}
\begin{table}[ht]
\centering
\caption{\textcolor{black}{Results for the OA-SMPC with different values of maximum probability of violation of state constraints ($\alpha$) for the year 2019.} 
}\vspace{-0.8em}
\resizebox{0.7\columnwidth}{!}{\begin{tabular}{lrrrr}
\cline{1-4}
Costs            &  $\alpha=0.05$   &  $\alpha=0.15$ &  $\alpha=0.2$ \\ 
\cline{1-4}
NCDC             & \$66,294     & \$59,564 &  \$57,323     \\
OPDC             & \$1,865      & \$1,036  &  \$540      \\
Energy Cost      & \$14,383     & \$14,395 &  \$14,395     \\
BESS loss        & \$8,473      & \$9,747  &  \$10,336      \\
Total Cost       & \$91,014     & \$84,742 &  \$82,593     \\
Total BESS cycles &169.5        &194.9         &206.7      \\
$Y$ at year end    & 4.5\%          & 14.3\%    & 19.1\%      \\
\cline{1-4}\\
\end{tabular}}
\label{Year_Table_alpha}
\vspace{-2em}
\end{table}

\begin{figure}[ht]
\includegraphics[scale=0.082]{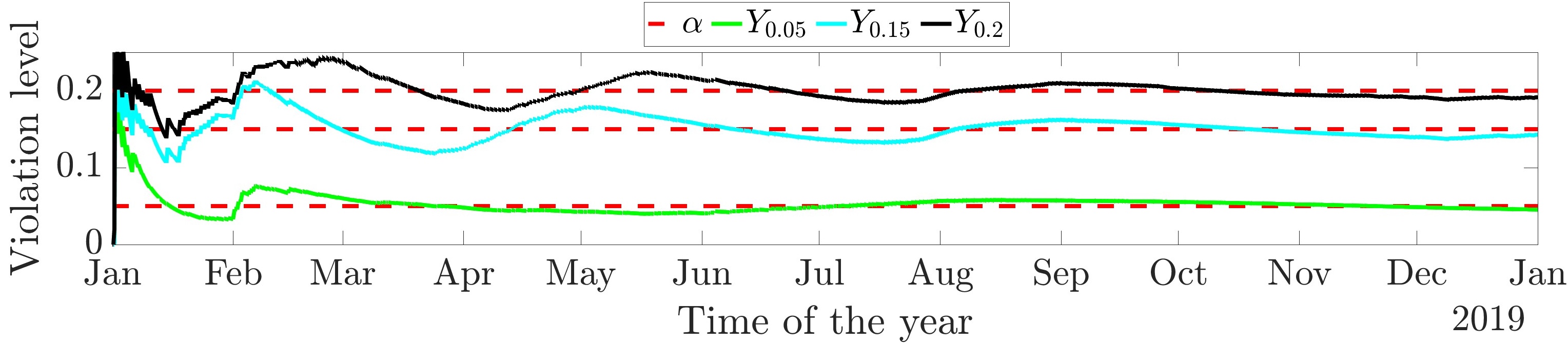}
\vspace{-2em}
\caption{\textcolor{black}{{Yearly time-series for the time-average of state constraint violations ($Y_{\alpha}$) for the OA-SMPC for the maximum allowable state constraint violation probability $\alpha\in\{0.05,0.15,0.2\}.$}}}
\vspace{-1em}
\label{Year_Figure_alpha}
\end{figure}

\section{Conclusions and Future Work}\label{conclusions}

This work presents a novel online adaptive state constraint relaxation based stochastic MPC (OA-SMPC) framework for non-conservative chance constraint satisfaction in closed-loop. An adaptive state constraint relaxation rule is developed for a generic discrete LTI system based on the time-average of past constraint violations without any a-priori assumptions about the probability distribution of the uncertainty set or its statistics, or sample uncertainties from historical data. The time-average of the state constraint violations, under the assumption of ideal control inputs which can cause/prevent constraint violations almost surely, is proven to asymptotically 
converge to the maximum allowable violation probability. \textcolor{black}{The time-average of the state constraint violations is also proven to exhibit martingale-like behavior asymptotically, even without the ideal control input assumption.}  

\textcolor{black}{The proposed method (OA-SMPC) is applied for minimizing monthly electricity costs by optimal BESS dispatch for a grid connected microgrid (MG) in an economic MPC (EMPC) framework. We perform simulations for the Port of San Diego MG using realistic PV and load forecast data for the year 2019. Chance constraints are applied on the BESS SOC to make use of excess BESS capacity in our proposed OA-SMPC as compared to the traditional \textcolor{black}{EMPC without chance constraints} (which uses hard constraints on BESS SOC). 
The OA-SMPC outperforms the traditional \textcolor{black}{EMPC and a state-of-the-art chance constrained approach from the literature}, striking an effective trade off between high BESS utilization (i.e., higher cost savings) and full SOC constraint satisfaction (i.e., longer BESS lifetime). The OA-SMPC lowers MG electricity costs by non-conservative chance constraint satisfaction in closed-loop, thereby having minimal adverse effect on BESS lifetime.} Future work will incorporate the BESS degradation cost and a life-cycle analysis of the MG.




\renewcommand{\theequation}{A.\arabic{equation}}
  \setcounter{equation}{0}  
  
{\appendix[
]\label{appendix_main}

\textcolor{black}{
A test for Assumption \ref{assumption_control_inputs}(a) entails 
solving the following optimization problem~\cite[Section 5.8.1]{boyd2004convex} at each time $t$, given $x(t)$, $c(t|t)$ and $h(t)$
.} 


\vspace{-1.5em}
\textcolor{black}{
\begin{equation}
\label{eq_1}
\begin{aligned}
{f^*} = {\min}\ {0},
\end{aligned}
\end{equation}
\vspace{-1.5em}
subject to
\begin{equation} \label{eq_1_1}
\begin{aligned}
x(t+1|t) = Ax(t)+Bu(t|t),
\end{aligned}
\end{equation}
\vspace{-1.5em}
\begin{equation} \label{eq_2}
\begin{aligned}
Su(t|t) \leq s,
\end{aligned}
\end{equation}
\vspace{-1.5em}
\begin{equation} \label{eq_2_2}
\begin{aligned}
Mu(t|t) =c(t|t), 
\end{aligned}
\end{equation}
\vspace{-1.5em}
\begin{equation} \label{eq_3}
\begin{aligned}
G(x(t+1|t)) \leq g-h(t).
\end{aligned}
\end{equation}
If $f^* = 0$, 
then the optimization problem defined above is feasible and we can guarantee Assumption \ref{assumption_control_inputs}(a) holds. If $f^* = \infty$, then the optimization problem above is infeasible and Assumption \ref{assumption_control_inputs}(a) fails to hold. Note that if~\eqref{equal_real} and~\eqref{equality_compact_nominal} are not considered in the problem formulation, we can drop~\eqref{eq_2_2}.}

}

\section*{\textcolor{black}{Acknowledgments}}
\textcolor{black}{The authors would like to thank Dr. Sonia Mart\'{\i}nez, Professor of Mechanical and Aerospace Engineering at UC San Diego, for the fruitful discussions during writing the manuscript. The authors would also like to extend their deepest gratitude to the anonymous reviewers whose comments greatly improved the content of the paper.}

\bibliographystyle{IEEEtran}
\bibliography{IEEE_Journal}

\end{document}